\documentclass[12pt]{article}
\usepackage{amsmath}
\usepackage{graphicx}

\usepackage{amsfonts,amssymb,amsmath,amsthm}
\usepackage{hyperref}
\usepackage{graphicx}
\usepackage{comment}
\usepackage{enumerate}

\makeatletter
\newtheorem*{rep@theorem}{\rep@title}
\newcommand{\newreptheorem}[2]{%
\newenvironment{rep#1}[1]{%
 \def\rep@title{#2 \ref{##1}}%
 \begin{rep@theorem}}%
 {\end{rep@theorem}}}
\makeatother

\newtheorem{theorem}{Theorem} 
 
\newtheorem{prop}{Proposition}  
\newtheorem{definition}{Definition}  
  
\newtheorem{claim}{Claim}  
\newtheorem{lemma}{Lemma}

\newtheorem{corol}{Corollary}

\newreptheorem{theorem}{Theorem}
\newreptheorem{lemma}{Lemma}
\newreptheorem{prop}{Proposition}

\newcommand{\<}{\langle}
\renewcommand{\>}{\rangle}
\newcommand{\be}{\begin{equation} }
\newcommand{\ee}{\end{equation} }
\newcommand{\ba}{\begin{eqnarray} }
\newcommand{\ea}{\end{eqnarray} }

\newcommand{\bpm}{\begin{pmatrix}}
\newcommand{\epm}{\end{pmatrix}}
\newcommand{\bmm}{\begin{matrix}}
\newcommand{\emm}{\end{matrix}}

\title{Constraints on order and disorder parameters in quantum spin chains}
\author{Michael Levin\thanks{Department of Physics, James Franck Institute, University of Chicago, Chicago, Illinois 60637, USA.}}

\begin{document}

\maketitle
\begin{abstract}
We derive general constraints on order and disorder parameters in Ising symmetric spin chains. Our main result is a theorem showing that every gapped, translationally invariant, Ising symmetric spin chain has either a nonzero order parameter or a nonzero disorder parameter. We also prove two more constraints on order and disorder parameters: (i) it is not possible for a gapped, Ising symmetric spin chain to have \emph{both} a nonzero order parameter and a nonzero disorder parameter; and (ii) it is not possible for a spin chain of this kind to have a nonzero disorder parameter that is \emph{odd} under the symmetry. These constraints have an interesting implication for self-dual Ising symmetric spin chains: every self-dual spin chain is either gapless or has a degenerate ground state in the thermodynamic limit. All of these constraints generalize to spin chains without translational symmetry. Our proofs rely on previously known bounds on entanglement and correlations in one dimensional systems, as well as the Fuchs-van de Graaf inequality from quantum information theory. 
\end{abstract}

\newpage

\tableofcontents

\newpage

\section{Introduction}

There are two distinct ways to probe symmetry breaking in quantum many-body systems. One approach is to use order parameters --- local operators that transform under some non-trivial representation of the symmetry group. The other approach is to use disorder parameters --- non-local operators that implement a symmetry transformation in some region, dressed with local operators along the boundary of the region\cite{Fradkin17}. These two approaches are complementary: while order parameters are useful for detecting symmetry breaking, disorder parameters are useful for detecting the \emph{absence} of symmetry breaking. 

The goal of this paper is to uncover the precise relationship between these two types of probes. To make progress, we focus on the simplest class of many-body systems in which order and disorder parameters can be defined: one dimensional quantum spin chains with Ising symmetry. 

We begin by recalling the (rough) definitions of order and disorder parameters in the context of Ising symmetric spin chains. Consider a spin chain composed out of spin-$1/2$ moments. Suppose that the Hamiltonian commutes with the Ising symmetry transformation, $S = \prod_i \sigma^x_i$. Such a spin chain has a nonzero ``order parameter'' if there exists an operator $O_i$ that is localized near site $i$, is odd under $S$ (i.e. satisfies $S O_i S = - O_i$), and obeys 
\begin{align*}
\lim_{|i-j| \rightarrow \infty} \<O_i^\dagger O_j\> \neq 0,
\end{align*}
where $\<\cdot\>$ denotes the ground state expectation value. Likewise, such a spin chain has a nonzero ``disorder parameter'' if there exists an operator $O_i$ that is localized near site $i$, is either even or odd under $S$ (i.e. satisfies $S O_i S = \pm O_i$), and obeys 
\begin{align*}
\lim_{|i-j| \rightarrow \infty} \<O_i^\dagger O_j \prod_{k=i+1}^j \sigma_k^x\> \neq 0.
\end{align*} 
(Here, $\prod_{k=i+1}^j \sigma_k^x$ implements the Ising symmetry transformation within the interval $[i+1,j]$).

The canonical example is the transverse field Ising model:
\begin{align*}
H = -\sum_i \sigma^z_i \sigma^z_{i+1} - B \sum_i \sigma^x_i,
\end{align*}
Depending on the value of $B$, this model can be in two phases: a phase with spontaneously broken symmetry for $|B| < 1$, and a phase without symmetry breaking for $|B| > 1$. These two phases illustrate the concepts of order and disorder parameters. Indeed, a well-known property of the symmetry breaking phase is that $\lim_{|i-j| \rightarrow \infty} \<\sigma_i^z \sigma_j^z\> \neq 0$. By the above definition, this property means that the symmetry breaking phase has a nonzero order parameter, namely $O_i = \sigma^z_i$. Likewise, a well-known property of the symmetric phase is that $\lim_{|i-j| \rightarrow \infty} \<\prod_{k=i+1}^j \sigma_k^x\> \neq 0$.\footnote{One way to derive this result is to note that the Kramers-Wannier duality transformation maps the operator $\sigma_i^z \sigma_j^z$ to $\prod_{k=i}^{j-1} \sigma^x_k$ (see section \ref{sec-self-dual}).} Following the above definition, this means that the symmetric phase has a nonzero disorder parameter, namely $O_i = \mathbb{1}$.

The above example raises the central question of this paper: we can see that the transverse field Ising model has \emph{either} a nonzero order parameter or a nonzero disorder parameter for every $B$ except at the critical points $B= \pm 1$. But is this true more generally? That is, does every gapped, translationally invariant, Ising symmetric spin chain have either a nonzero order parameter or a nonzero disorder parameter?

Conventional wisdom suggests that the answer is ``yes.'' The reasoning goes as follows: we expect that every spin chain of this kind can be adiabatically connected to the transverse field Ising model in either the symmetry breaking ($|B| < 1$) or symmetric ($|B| > 1$) phase\cite{NorbertSPT, ChenGuWen}. Therefore, since both phases of the transverse field Ising model have a nonzero order parameter or a nonzero disorder parameter, every gapped spin chain must also have this property\cite{HastingsWen}.

In this paper, we confirm this intuition: we show that every gapped, translationally invariant, Ising symmetric spin chain has either a nonzero order parameter or a nonzero disorder parameter (Theorem \ref{mainthm1}). Our proof, however, does not follow the above line of reasoning, which is difficult to make rigorous. Instead, we derive a \emph{direct} connection between order and disorder parameters using results from quantum information theory\cite{Fuchs} together with bounds on correlations\cite{Hastings_area_law, Hamzaetal} and entanglement\cite{AKLV} in one dimensional gapped ground states. 

Our arguments also generalize to spin chains without translational symmetry. In that case, the statement of our theorem is slightly different since there is no reason that there has to be a nonzero order parameter or nonzero disorder parameter that is defined throughout the entire spin chain --- parts of the spin chain may be ordered and other parts may be disordered. Thus, what we prove is that the spin chain can always be divided into a few regions, each of which supports either a nonzero order parameter or a nonzero disorder parameter (Theorem \ref{mainthm2}).

In addition to the above constraint (the main result of this paper), we also prove two other general constraints on order and disorder parameters: we show that (i) it is not possible for a gapped, Ising symmetric spin chain to have \emph{both} a nonzero order parameter and a nonzero disorder parameter; and (ii) it is not possible for a spin chain of this kind to have a nonzero disorder parameter in which $O_i$ is \emph{odd} under the symmetry (Theorems \ref{mutconst}-\ref{exchconst}). All together, we believe that this is a complete list of constraints on order and disorder parameters in Ising symmetric spin chains.

As a bonus, our constraints have an interesting implication for \emph{self-dual} spin chains that are invariant under the Kramers-Wannier duality transformation: they imply that every self-dual Ising symmetric spin chain is either gapless or has a degenerate ground state in the thermodynamic limit (Corollary \ref{selfdual}).

What is our motivation for considering these questions? First, order and disorder parameters are key concepts in the theory of symmetry breaking in one dimension, and therefore we believe it is intrinsically important to understand their relationship to one another. In fact, it is possible to generalize order and disorder parameters to higher dimensional spin systems and gauge theories\cite{Fradkin17, thooft78} and one can ask similar questions in that context; the questions we consider here are the simplest examples of a larger class of problems. 

Another source of motivation has to do with a seemingly unrelated problem about gapped boundaries of two dimensional topological phases. Previous work has argued that gapped boundaries of two dimensional Abelian topological phase always obey a property known as ``braiding non-degeneracy''\cite{levin2013protected}. This property has important physical implications. For example, it is key for understanding why the boundaries of certain topological phases are necessarily gapless\cite{levin2013protected}. With Theorem \ref{mainthm1}, we can rigorously establish this property for the simplest nontrivial topological phase: the toric code model\cite{kitaevtc, BravyiKitaev}. This application will be explored in a follow-up paper\cite{Levinprep}.

This paper is organized as follows. In section \ref{sec-main-result}, we discuss our main result: every gapped, translationally invariant, Ising symmetric spin chain has either a nonzero order parameter or a nonzero disorder parameter. In section \ref{sec-constraints} we present two other constraints on order and disorder parameters (namely constraints (i) and (ii) mentioned above) and in section \ref{sec-self-dual} we discuss the implications of these constraints for self-dual spin chains. In sections \ref{sec-main-theorem} and \ref{sec-two-more-theorems} we present the proofs of our constraints.

\begin{figure}[tb]
\includegraphics[width=.5\columnwidth]{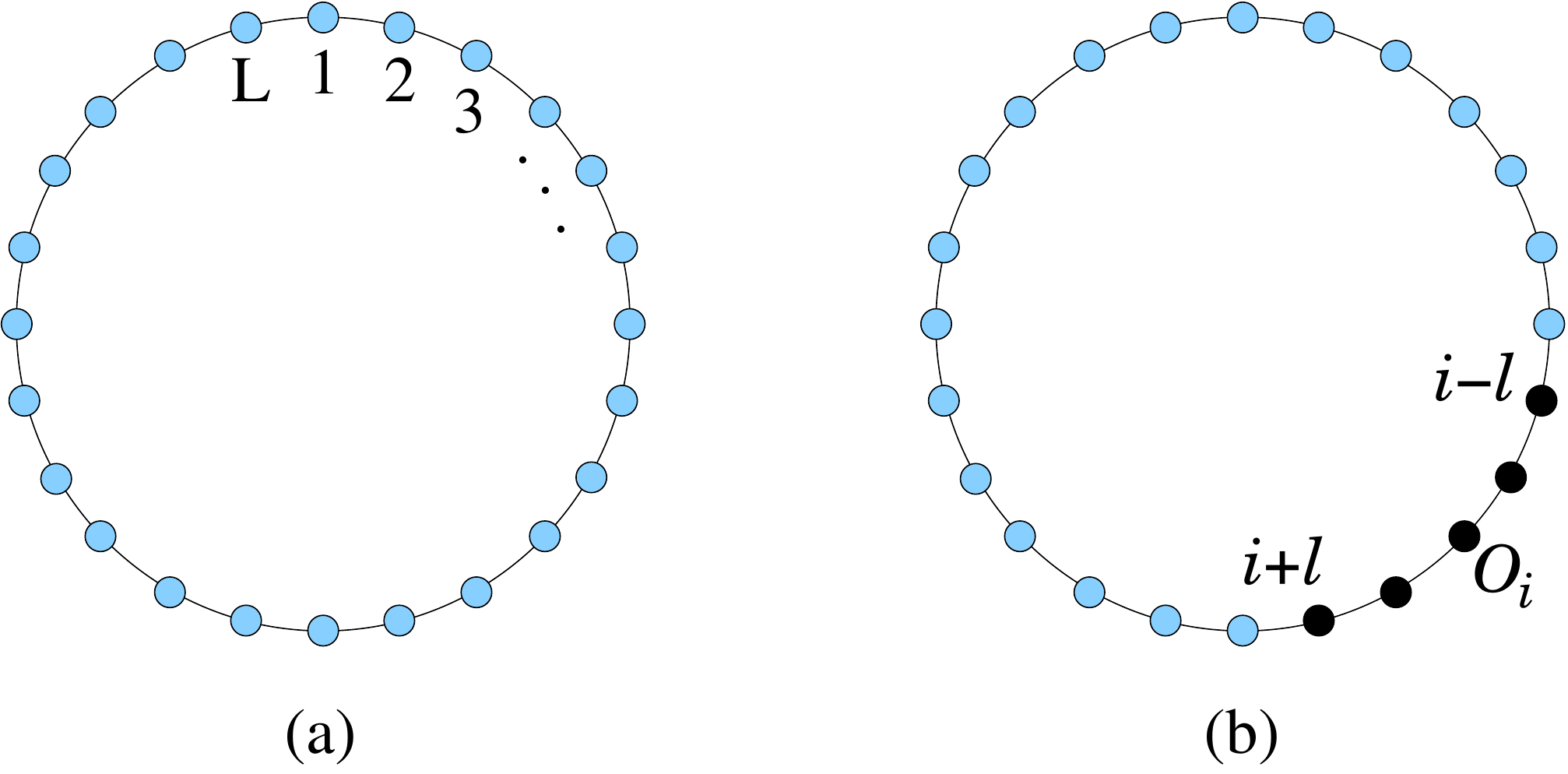}
\centering
\caption{(a) A spin chain made up of $L$ spins in a ring geometry. (b) A $(\delta, \ell)$ order parameter is a collection of operators $\{O_i\}$, supported on the intervals $[i-\ell, i+\ell]$, obeying certain conditions. }
\label{Hoifig}
\end{figure}

\section{Main result}
\label{sec-main-result}
Before stating our main result, we need to explain our setup. We consider spin chains made up of $L$ spins, arranged in a ring geometry (Fig. \ref{Hoifig}a). Each spin has $d$ possible states. We consider Hermitian Hamiltonians with nearest neighbor interactions and with bounded norm:
\begin{align}
H = \sum_{i=1}^L H_{i,i+1}, \quad \quad \|H_{i,i+1}\| \leq 1, \quad \label{ham}
\end{align}
We assume that $H$ is invariant under a (unitary) Ising symmetry transformation of the form 
\begin{align}
S = \prod_{i=1}^L S_i, \quad \quad S_i^2 = 1
\end{align}
where $S_i$ is a unitary operator acting on each $d$-dimensional spin. Note that we do not sacrifice any generality in restricting to Hamiltonians with nearest-neighbor interactions: every one dimensional Hamiltonian with finite range interactions can be represented using nearest neighbor interactions by clustering sufficiently large groups of spins into superspins. Also note that we do not assume that $H$ is translationally invariant at this point.

A few comments about notation: {\bf (1)} We denote the expectation value of the operator $O$ in the state $|\psi\>$ by $\<O\>_\psi \equiv \<\psi|O |\psi\>$; {\bf (2)} We say that an operator $O$ is ``even under $S$'' if $S O S = O$, and ``odd under $S$'' if $S O S = - O$. Likewise, we will say that a state $|\psi\>$ in the spin-chain Hilbert space is ``even under $S$'' if $S |\psi\> = |\psi\>$, and ``odd under $S$'' if $S |\psi\> = - |\psi\>$; {\bf (3)} We define the ``distance'' between sites $i, j$ in the obvious way: $\mathrm{dist}(i,j) = \min(|i-j|, L-|i-j|)$.

Next we define the notion of an order parameter --- or more precisely a ``$(\delta,  \ell)$ order parameter'':
\begin{definition}\label{deford}
A collection of operators $\{O_i: i \in X \}$ with $X \subset \{1,...,L\}$ is called a $(\delta,\ell)$ order parameter for a state $|\psi\>$, if 
\begin{enumerate}
\item{$O_i$ is odd under $S$.}
\item{$O_i$ is supported on $[i - \ell, i + \ell]$.}
\item{$|\<O_i^\dagger O_{j}\>_\psi| \geq \delta \text{ for all } i,j \in X \text{ with }\text{dist}(i,j) > 2\ell.$}
\item{$\|O_i\| \leq 1$.}
\end{enumerate}
The set $X$ is called the domain of definition of the order parameter.
\end{definition}
Here $\delta$ and $\ell$ are two positive real numbers that characterize different aspects of the order parameter: $\ell$ describes the size of the region of support of the $O_i$ operators (Fig. \ref{Hoifig}b) while $\delta$ describes the strength of the (two-point) correlations. Likewise, the subset $X \subset \{1,...,L\}$ should be thought of as the region where the order parameter is defined. In many cases, order parameters will be defined on the whole system, $X = \{1,...,L\}$, but the above definition allows for the more general possibility where the order parameter is only defined on a subset of the system. This is particularly relevant to systems without translational symmetry.

Our definition for a \emph{disorder} parameter is similar: 
 \begin{definition}
A collection of operators $\{O_i: i \in X \}$ is called a $(\delta,\ell)$ disorder parameter for a state $|\psi\>$ 
if 
\begin{enumerate}
\item{The $O_i$ operators are all even under $S$ or all odd under $S$.}
\item{$O_i$ is supported on $[i - \ell, i + \ell]$.}
\item{$|\<O_i^\dagger O_{j} \prod_{k=i+1}^{j} S_k\>_\psi| \geq \delta \text{ for all } i,j \in X \text{ with }\text{dist}(i,j) > 2\ell.$}
\item{$\|O_i\| \leq 1$.}
\end{enumerate}
The set $X$ is called the domain of definition of the disorder parameter.
\end{definition}
Note that in the definition of disorder parameters, the $O_i$ operators are either all even or all odd. Thus we can distinguish between \emph{two} types of disorder parameters: ``even'' and ``odd'' disorder parameters. 

We are now ready to present our main result. Let $H$ be an Ising symmetric Hamiltonian of the form given in Eq. (\ref{ham}). The eigenstates of $H$ can always be chosen so that they have a definite parity --- even or odd --- under the Ising symmetry. In this paper, we focus on the \emph{even} eigenstates.\footnote{Of course, all of our results also hold if we replace ``even'' by ``odd.''} Our results depend on a single assumption: the lowest energy eigenvalue within the even subspace is \emph{non-degenerate}. We will denote the corresponding eigenstate by $|\Omega\>$ and the energy gap within the even subspace by $\epsilon$.
Our first result is that if $H$ is translationally invariant, then $|\Omega\>$ is guaranteed to have either a $(\delta, \ell)$ order parameter or a $(\delta, \ell)$ disorder parameter for appropriate $\delta$ and $\ell$:
\begin{theorem}
\label{mainthm1}
If $H$ is translationally invariant, then $|\Omega\>$ has either a $(\delta, \ell)$ order parameter or a $(\delta, \ell)$ disorder parameter defined on the whole spin chain (i.e. $X = \{1,...,L\}$) with
\begin{align}
\delta = \frac{1}{72}, \quad \quad \ell \leq \tilde{\mathcal{O}} \left(\frac{(\log d)^3}{\epsilon^2} \right)
\label{delta0l0}
\end{align}
\end{theorem}
(Here, and in the remainder of the paper, we use $\tilde{\mathcal{O}}$ notation as follows: $f(x) = \tilde{\mathcal{O}}(g[x])$ means that there exists constants $C$ and $k$ such that $|f(x)| \leq C g[x] (\log g[x])^k$ for all $x$). 

A few comments about Theorem \ref{mainthm1}: 

{\bf (1)} It is important to recognize that the existence of a $(\delta, \ell)$ order/disorder parameter is only meaningful if $\ell$ is smaller than the system size $L$. This means that the theorem does not give any meaningful constraints on gapless spin chains whose energy gaps scale like $\epsilon \sim L^{-1}$ since in that case the bound on $\ell$ is larger than $L$. On the other hand, the theorem does tell us something about \emph{gapped} spin chains, i.e. spin chains where $\epsilon$ is bounded away from zero for arbitrarily large system sizes $L$; in this case, the theorem implies that there must be either an order parameter or a disorder parameter supported within a finite length scale $\ell$ that does not scale with the system size. 

{\bf (2)} The reader may worry that the theorem is inconsistent with the standard picture in which order/disorder parameters become arbitrarily small near critical points. The key point is that, unlike the standard setup, we do not fix a choice of order/disorder parameters across different Hamiltonians: these operators can change as we approach a critical point and, in particular, the region of support of these operators can grow (subject to the bound in Eq. \ref{delta0l0}) in order to maintain a fixed value of $\delta$.

{\bf (3)} To illustrate the theorem, consider the transverse field Ising model near the critical point. In this case, Theorem \ref{mainthm1} implies that there exists either a $(\delta, \ell)$ order parameter or $(\delta, \ell)$ disorder parameter with $\delta \sim 1$ and $\ell \lesssim \epsilon^{-2}$. It is interesting to compare this result to the usual renormalization group picture for the transverse field Ising model. Imagine we coarse grain the spin chain to a length scale of the same order as the correlation length $\xi$. According to the renormalization group picture, we expect that the usual order and disorder parameters will have strength $\delta \sim 1$ at this length scale. This suggests that we can always find an order or disorder parameter with $\delta \sim 1$ and $\ell \lesssim \xi \sim \epsilon^{-1}$. Comparing with the above results, we see that Theorem \ref{mainthm1} is consistent with the renormalization group picture and in fact the bounds on $\ell$ are weaker then one might expect. This suggests that it may be possible to strengthen these bounds further. \\

We now move on to the more general case where $H$ is not necessarily translationally invariant. In this case, the ground state can be spatially inhomogeneous so we may anticipate that the system can contain three types of regions: (1) regions that support an order parameter, (2) regions that support a disorder parameter that is even under $S$, and (3) regions that support a disorder parameter that is odd under $S$. The following theorem tells us that in fact, these three regions encompass the entire spin chain:
\begin{theorem}
\label{mainthm2}
For general $H$, there exists three sets $X_o, X_{d1}, X_{d2}$ with $X_o \cup X_{d1} \cup X_{d2} = \{1,...,L\}$ such that $|\Omega\>$ has a $(\delta, \ell)$ order parameter on $X_o$, a $(\delta, \ell)$ disorder parameter that is even under $S$ on $X_{d1}$, and a $(\delta, \ell)$ disorder parameter that is odd under $S$ on $X_{d2}$. Here $\delta = \frac{1}{72}$ and $\ell \leq \tilde{\mathcal{O}}(\frac{(\log d)^3}{\epsilon^2})$.
\end{theorem}
Note that Theorem \ref{mainthm2} is the natural generalization of Theorem \ref{mainthm1}: it shows that general spin chains can always be divided into three regions, each of which supports either an order parameter or a disorder parameter with definite parity.

\section{Two more constraints on order and disorder parameters}
\label{sec-constraints}
In addition to the above results we also prove two more constraints on order and disorder parameters. These constraints do not simplify significantly in the translationally invariant case, so we will only state them in the general (non-translationally invariant) case. The first constraint is as follows:
\begin{theorem}
\label{mutconst}
If $|\Omega\>$ has a $(\delta, \ell)$ order parameter defined at two points $i_1, i_3$ and a $(\delta, \ell)$ disorder parameter defined at two points $i_2, i_4$ with $i_1 < i_2 < i_3 < i_4$, then
\begin{align}
\min_{k,l}[\mathrm{dist}(i_k,i_l)] \leq 2\ell + \tilde{\mathcal{O}}\left(\frac{\log \delta^{-1}}{\epsilon} \right)  
\label{mutineq}
\end{align}
\end{theorem}
To understand the implications of this result, consider the case where the $(\delta, \ell)$ order parameter and $(\delta, \ell)$ disorder parameter are defined throughout the \emph{entire} spin chain, i.e. $X = \{1,...,L\}$. In this case, we can choose $i_1, i_2, i_3, i_4$ to be equally spaced with a separation of $L/4$. Then Eq. \ref{mutineq} implies
\begin{align}
L \leq 8 \ell +  \tilde{\mathcal{O}}\left(\frac{\log \delta^{-1}}{\epsilon} \right)  
\label{Lbound}  
\end{align}
This bound implies that is \emph{not} possible for a spin chain to have both a $(\delta, \ell)$ order parameter and a $(\delta, \ell)$ disorder parameter defined throughout the whole chain and also have a finite gap $\epsilon$ in the limit $L \rightarrow \infty$. This is the constraint (i) mentioned in the introduction.

Our second constraint has a similar structure:
\begin{theorem}
\label{exchconst}
If $|\Omega\>$ has a $(\delta, \ell)$ disorder parameter that is odd under $S$, defined at $i_1, i_2, i_3, i_4$, then
\begin{align}
\min_{k,l}[\mathrm{dist}(i_k,i_l)] \leq 2\ell +  \tilde{\mathcal{O}}\left(\frac{\log \delta^{-1}}{\epsilon} \right)    
\label{exchineq}
\end{align}
\end{theorem}
Again, the meaning of this constraint is most clear in the case where there is an odd $(\delta, \ell)$ disorder parameter that is defined throughout the entire spin chain. As above, if we choose $i_1, i_2, i_3, i_4$ to be equally spaced we derive the inequality (\ref{Lbound}). It then follows that it is \emph{not} possible for a spin chain to have a $(\delta, \ell)$ disorder parameter that is odd under the symmetry and also have a finite gap $\epsilon$ in the limit $L \rightarrow \infty$. This is the constraint (ii) mentioned in the introduction.

Unlike Theorems \ref{mainthm1}-\ref{mainthm2}, the proofs of Theorem \ref{mutconst} and \ref{exchconst} are straightforward: both theorems follow from the non-trivial commutation algebra obeyed by order and disorder parameters. In particular, Theorem \ref{mutconst} originates from the fact that order and disorder parameters \emph{anticommute}, e.g. $\sigma_{i_1}^z \sigma_{i_3}^z$ anticommutes with $\prod_{k=i_2}^{i_4} \sigma_k^x$ if $i_1 < i_2 < i_3 < i_4$. Likewise, Theorem \ref{exchconst} comes from the fact that odd disorder parameters obey a \emph{fermionic} commutation algebra\cite{Levinwenferm}. We present these proofs in section \ref{sec-two-more-theorems}.

\section{Application to self-dual spin chains}
\label{sec-self-dual}
Our results have an interesting application to \emph{self-dual} spin chains. Here we define the notion of self-dual spin chains as follows. Consider a spin-$1/2$ chain consisting of $L$ spins arranged in a ring, and consider a Hamiltonian of the form $H = \sum_i H_i$ where $\|H_i\| \leq 1$ and where each $H_i$ is supported on the interval $[i,i+r]$ for some integer $r \leq L-1$, which we will refer to as the ``range'' of $H$. Suppose further that $H$ commutes with the Ising symmetry $S = \prod_{i=1}^L \sigma_i^x$. We say that $H$ is self-dual if $U_D^{\dagger} H_{\text{even}} U_D = H_{\text{even}}$ where $H_{\text{even}}$ is the restriction of $H$ to the even ($S=1$) subspace and $U_D$ is the unitary operator, defined within the even subspace, given by
\begin{align*}
U_D^{\dagger} \sigma^x_i U_D = \sigma^z_{i-1} \sigma^z_{i}, \quad \quad \quad
U_D^{\dagger} (\sigma^z_{i-1} \sigma^z_{i}) U_D  = \sigma^x_{i-1}
\end{align*}
The simplest example of a self-dual spin chain is the transverse field Ising model at $B=1$.

\begin{corol}
For a self-dual spin chain with interactions of range $r$, the energy gap $\epsilon$ within the even subspace is bounded by $\epsilon \leq \mathcal{O}(\frac{r^3}{L^{1/2-\alpha}})$ for every $\alpha > 0$.
\label{selfdual}
\end{corol}
(Here, and in the rest of the paper, $f(x) = \mathcal{O}(g(x))$ means that there exists a constant $C$ such that $|f(x)| \leq C g(x)$ for all $x$). 

This corollary implies that self-dual Ising symmetric spin chains cannot have a gapped, non-degenerate ground state (within the even subspace) in the thermodynamic limit $L \rightarrow \infty$. 
\begin{proof}
Let $|\Omega\>$ be the lowest energy even eigenstate. Given that the spin chain is self-dual we know that $U_D |\Omega\> \propto |\Omega\>$. To understand the implications of this, notice that the duality transformation maps order parameters onto \emph{even} disorder parameters and vice-versa, e.g.
\begin{align*}
U_D^\dagger \left(\prod_{k=i+1}^j \sigma^x_k \right) U_D = \sigma^z_i \sigma^z_j, \quad \quad \quad
U_D^\dagger (\sigma^z_i \sigma^z_j) U_D = \prod_{k=i}^{j-1} \sigma^x_k
\end{align*}
It follows that if $|\Omega\>$ has a $(\delta,\ell)$ order parameter, then it also has an even $(\delta,\ell)$ disorder parameter and vice-versa. 

Next observe that every self-dual Hamiltonian is translationally invariant since the square of the duality transformation acts like a unit translation on all even operators. This means that we can apply Theorem \ref{mainthm1} which tells us that $|\Omega\>$ has either a $(\delta, \ell)$ order parameter or a $(\delta, \ell)$ disorder parameter with
\begin{align}
\delta = \frac{1}{72}, \quad \quad \ell \leq \tilde{\mathcal{O}}\left(\frac{r^3}{(\epsilon/r)^2}\right)
\label{ellbound}
\end{align}
Here, the bound on $\ell$ comes from the fact that, in order to apply Theorem \ref{mainthm1}, we first need to cluster $r$ adjacent spin-$1/2$'s into superspins of dimension $d = 2^r$ in order to convert our spin chain into one with nearest neighbor interactions. In addition to clustering, we also have to multiply the resulting superspin Hamiltonian by $1/r$ in order to satisfy the norm condition (\ref{ham}); this leads to a gap of $\epsilon/r$ and explains the factor of $\epsilon/r$ in the denominator of (\ref{ellbound}).

Combining this observation with the previous one, we conclude that either (a) $|\Omega\>$ has \emph{both} a $(\delta, \ell)$ order parameter and an even $(\delta, \ell)$ disorder parameter obeying (\ref{ellbound}) or (b) $|\Omega\>$ has an odd $(\delta, \ell)$ disorder parameter obeying (\ref{ellbound}). In the first case, Theorem \ref{mutconst} implies that $\bar{L} \leq 8\ell + \tilde{\mathcal{O}}(\frac{r}{\epsilon})$ where $\bar{L} \equiv \frac{L}{r}$ is the length of our spin chain after clustering together $r$ adjacent spin-$1/2$'s into superspins. In the second case, we again deduce that $\bar{L} \leq 8\ell + \tilde{\mathcal{O}}(\frac{r}{\epsilon})$ --- this time using Theorem \ref{exchconst}. 

The last step is to use the bound on $\ell$ (\ref{ellbound}), which gives the inequality $L \leq \tilde{\mathcal{O}}(\frac{r^6}{\epsilon^2})$. It then follows from our definition of $\tilde{\mathcal{O}}$ that $L \leq \mathcal{O}([\frac{r^6}{\epsilon^2}]^{1+2\alpha})$ for every $\alpha > 0$. Inverting this inequality gives the desired bound $\epsilon \leq \mathcal{O}(\frac{r^3}{L^{1/2-\alpha}})$.
\end{proof}  

\section{Proof of main result}
\label{sec-main-theorem}
In this section, we prove Theorems \ref{mainthm1} and \ref{mainthm2}. We begin with some definitions.

\subsection{Definitions}
First some notation: for each subset $X \subset \{1,...,L\}$, let $S_{X}$ be the symmetry transformation restricted to $X$, i.e. 
\begin{align*}
S_{X} \equiv \prod_{i \in X} S_i.
\end{align*}
Next, we define a notion of when a state $|\psi\>$ is ``$\delta$ strongly-ordered'' on a pair of intervals $I_1, I_2$: 
\begin{definition}\label{strongord}
A state $|\psi\>$ is $\delta$ strongly-ordered on disjoint intervals $I_1, I_2$ if there exist operators $O_1$ and $O_2$ such that:
\begin{enumerate}
\item{$O_1$ and $O_2$ are supported on $I_1$ and $I_2$.}
\item{$O_1$ and $O_2$ are odd under $S$.}
\item{$|\<O_1 O_2\>_\psi| \geq \delta$.~\footnote{We use the condition $|\<O_1 O_2\>_\psi| \geq \delta$ instead of $|\<O_1^\dagger O_2\>_\psi| \geq \delta$ for two reasons: (1) the first condition is simpler and (2) the two conditions are equivalent since one can replace $O_1 \rightarrow O_1^\dagger$.}}
\item{$\|O_1\|, \|O_2\| \leq 1$.}
\end{enumerate}
\end{definition}
We also define a notion of when a state is ``$\delta$ \emph{weakly}-ordered'' on a pair of intervals $I_1, I_2$:
\begin{definition}
A state $|\psi\>$ is $\delta$ weakly-ordered on disjoint intervals $I_1, I_2$ if there exists an operator $A$ such that: 
\begin{enumerate}
\item{$A$ is supported on $I_1 \cup I_2$.}
\item{$A$ is odd under $S_{I_1}$ and $S_{I_2}$.}
\item{$|\<A\>_\psi| \geq \delta$.}
\item{$\|A\| \leq 1$.}
\end{enumerate}
\end{definition}
To get some intuition about the above definition, notice that we can always expand $A$ as a linear combination $A = \sum_\alpha O_{1\alpha} O_{2\alpha}$ where $O_{1\alpha}, O_{2\alpha}$ are operators supported on $I_1, I_2$. The second condition then amounts to the requirement that $O_{1 \alpha}$ and $O_{2 \alpha}$ are odd under $S$ for all $\alpha$. Likewise, the third condition becomes $|\sum_\alpha \< O_{1\alpha} O_{2\alpha}\>_\psi| \geq \delta$. From this point of view, we can see that the main difference between $\delta$ weak order and $\delta$ strong order is that the former characterizes the correlations that can be detected by \emph{entangled} operators of the form $\sum_\alpha O_{1\alpha} O_{2\alpha}$, while the latter characterizes the correlations that can be detected by simple products $O_1 O_2$. 

We define the notion of $\delta$ strong disorder and $\delta$ weak disorder in a similar fashion:
\begin{definition}
A state $|\psi\>$ is $\delta$ strongly-disordered on disjoint, non-adjacent intervals $I_1, I_2$ if there exist operators $O_1$ and $O_2$ such that
\begin{enumerate}
\item{$O_1$ and $O_2$ are supported on $I_1$ and $I_2$.}
\item{$O_1$ and $O_2$ are both even or both odd under $S$.}
\item{$|\<O_1 O_2 S_{J}\>_\psi| \geq \delta$ where $J$ is the interval between $I_1$ and $I_2$, going clockwise.}
\item{$\|O_1\|, \|O_2\| \leq 1$.}
\end{enumerate}
\end{definition}
Notice that $O_1, O_2$ are either both even or both odd under $S$: in the former case, we will say that $|\psi\>$ is $\delta$ strongly-disordered with \emph{even} parity, while in the latter case we will say that $|\psi\>$ is $\delta$ strongly-disordered with \emph{odd} parity.

\begin{definition}
A state $|\psi\>$ is $\delta$ weakly-disordered on disjoint, non-adjacent intervals $I_1, I_2$ if there exists an operator $A$ such that
\begin{enumerate}\label{weakdisord}
\item{$A$ is supported on $I_1 \cup I_2$.}
\item{$A$ is even under both $S_{I_1}, S_{I_2}$ or odd under both $S_{I_1}, S_{I_2}$.}
\item{$|\<A S_{J}\>_\psi| \geq \delta$ where $J$ is the interval between $I_1$ and $I_2$, going clockwise}.
\item{$\|A\| \leq 1$.}
\end{enumerate}
\end{definition}
Similarly to above, we will say that $|\psi\>$ is $\delta$ weakly-disordered with even or odd parity depending on whether $A$ is even or odd under $S_{I_1}, S_{I_2}$.

\begin{figure}[tb]
\includegraphics[width=.25\columnwidth]{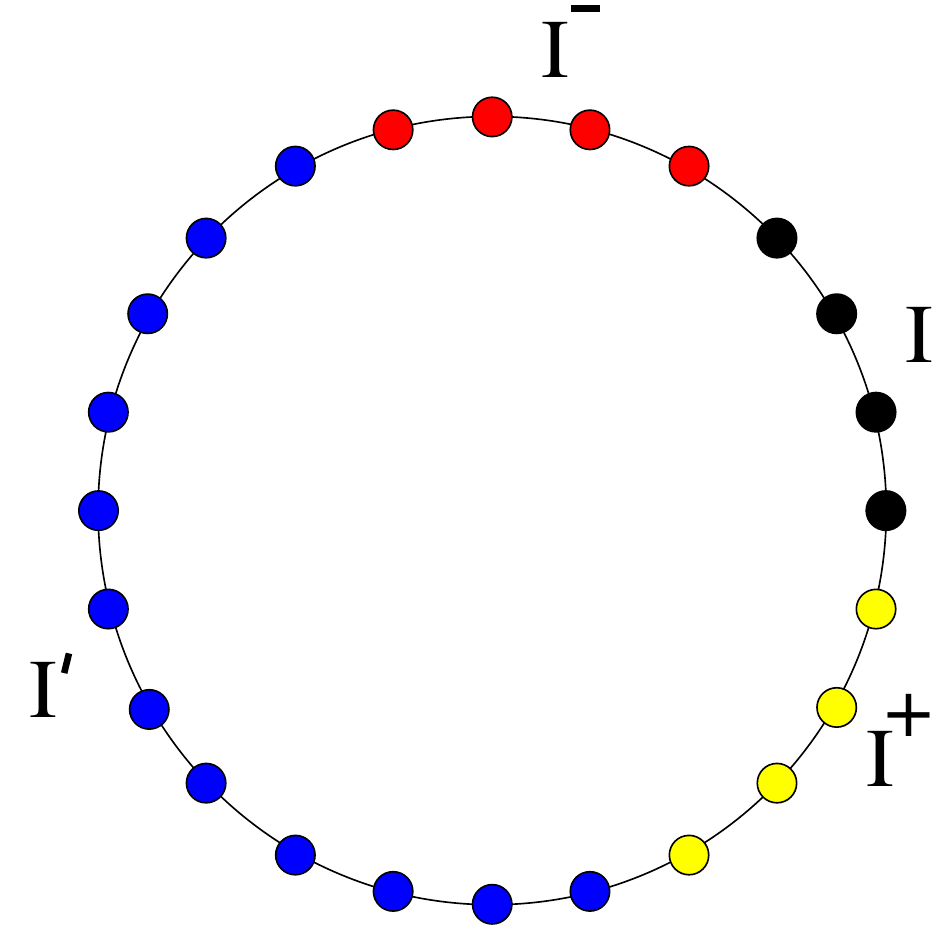}
\centering
\caption{Example of $I^\pm$ notation: the spins in the interval $I$ are shown in black; the spins in $I^-$ and $I^+$ are shown in red and yellow; and the
spins in $I' \equiv (I^- \cup I \cup I^+)^c$ are shown in blue.}
\label{Iplusminus}
\end{figure}

\subsection{Main argument}

Before giving the proof of Theorems \ref{mainthm1}-\ref{mainthm2}, we need to introduce one last piece of notation: given an interval $I \subset \{1,...,L\}$ we define $I^+$ and $I^-$ to be the two intervals that are adjacent to $I$ and have the same length as $I$ (see Fig. \ref{Iplusminus}). Also, we define $I'$ to be the interval 
\begin{align}
I' \equiv (I^- \cup I \cup I^+)^c.
\end{align}

\begin{proof} (of Theorems \ref{mainthm1} and \ref{mainthm2}) We start by proving Theorem \ref{mainthm2}: we show that there exists three sets 
$X_o, X_{d1}, X_{d2}$ with $X_o \cup X_{d1} \cup X_{d2} = \{1,...,L\}$, such that $|\Omega\>$ has a $(\delta, \ell)$ order parameter on $X_o$, a $(\delta, \ell)$ disorder parameter that is even under $S$ on $X_{d1}$, and a $(\delta, \ell)$ disorder parameter that is odd under $S$ on $X_{d2}$. 

We proceed in three steps --- each of which depends on an associated lemma. We give the proofs of these lemmas in the next section. The first lemma is a quantum information theory result that applies to \emph{any} state $|\psi\>$ that has a definite parity under $S$:
\begin{lemma}
\label{orddis}
Let $|\psi\>$ be a state that is even or odd under $S$. For every pair of disjoint, non-adjacent intervals $I_1, I_2$, the state $|\psi\>$ is either $\delta$ weakly-ordered on $I_1, I_2$ or $(1-\delta)/2$ weakly-disordered on the complementary intervals $J_1, J_2$, where the intervals are ordered $(J_1, I_1, J_2, I_2)$, going clockwise.
\end{lemma}
Lemma \ref{orddis} is useful because it gives a way to construct the three sets $X_o$, $X_{d1}$ and $X_{d2}$. Fix an integer $\ell > 0$, to be chosen later, and for each $i \in \{1,...,L\}$, define $I_i \equiv [i+1, i+\ell]$. We claim that the state $|\Omega\>$ is either $1/3$ weakly-ordered on $I_i, I_i'$ or $1/3$ weakly-disordered on $I_i^-, (I_i^-)'$ for each $i$. To see this, apply Lemma \ref{orddis} to the two intervals $I_i, I_i'$ with $\delta = 1/3$. The lemma tells us that $|\Omega\>$ is either $1/3$ weakly-ordered on $I_i, I_i'$ or $1/3$ weakly-disordered on the complementary intervals $I_i^-, I_i^+$. This implies our claim since every state that is $1/3$ weakly-disordered on $I_i^-, I_i^+$ is also $1/3$ weakly disordered on $I_i^-, (I_i^-)'$ since $I_i^+ \subset (I_i^-)'$. We now define:
\begin{align*}
X_o &= \{i: \text{$|\Omega\>$ is $1/3$ weakly-ordered on $I_i, I_i'$}\} \\
X_{d1} &= \{i: \text{$|\Omega\>$ is $1/3$ weakly-disordered on $I_i^-, (I_i^-)'$ with even parity}\} \\
X_{d2} &= \{i: \text{$|\Omega\>$ is $1/3$ weakly-disordered on $I_i^-, (I_i^-)'$ with odd parity}\}
\end{align*}
By construction, $X_o \cup X_{d1} \cup X_{d2} = \{1,...,L\}$.
To proceed further, we invoke the following lemma:
\begin{lemma}\label{weakimpstrong}
For every $\delta > 0$, there exists a length\footnote{Here, we treat $\delta$ as a constant when using $\tilde{\mathcal{O}}$ notation.} $\lambda = \tilde{\mathcal{O}}(\frac{(\log d)^3}{\epsilon^2})$ such that if $|\Omega\>$ is $\delta$ weakly-ordered on $I_1, I_2$ then it is $\delta/2$ strongly-ordered on $I_1, I_2$ for all $I_1, I_2$ that are separated by a distance of at least $\lambda$.  Likewise, if $|\Omega\>$ is $\delta$ weakly-disordered on $I_1, I_2$ with even (odd) parity then it is $\delta/2$ strongly-disordered on $I_1, I_2$ with even (odd) parity for all $I_1, I_2$ separated by at least $\lambda$.  
\end{lemma}
Lemma \ref{weakimpstrong} is important because it allows us to upgrade from \emph{weak} order/disorder to \emph{strong} order/disorder. This upgrade, in turn, allows us to construct order and disorder parameters. To construct an order parameter, we apply the above lemma with $\delta = 1/3$. We then choose $\ell$ to be larger than the length $\lambda$ given in the above lemma. With this choice of $\ell$, it follows that $|\Omega\>$ is $1/6$ strongly-ordered on $I_i, I_i'$ for each $i \in X_o$. This means that, for each $i \in X_o$, there exists operators $O_i, O_i'$ supported on $I_i, I_i'$ that are odd under the symmetry and satisfy $|\<O_i O_i'\>_\Omega| \geq 1/6$. We define our order parameter to be the corresponding $O_i$ operators: $\{O_i : i \in X_o\}$. Likewise, to construct \emph{disorder} parameters, we use the lemma to deduce that, for each $i \in X_{d1}$ ($i \in X_{d2}$), there exists operators $O_i, O_i'$ supported on $I_i^-, (I_i^-)'$ that are even (odd) under the symmetry and satisfy $|\<O_i O_i' S_{I_i}\>_\Omega| \geq 1/6$. We define the disorder parameters on $X_{d1}, X_{d2}$ to be the corresponding $O_i$ operators: $\{O_i : i \in X_{d1}\}$ and $\{O_i : i \in X_{d2}\}$. 

All that remains is to show that the above order and disorder parameters have long range correlations. We do this with the help of the following lemma:
\begin{lemma}\label{ordparam}
Let $I_1, I_2$ be two disjoint intervals and let $\ell_{\text{min}} = \text{min}(|I_1|, |I_2|, \text{dist}(I_1, I_2))$. For any operators $O_1, O_1', O_2, O_2'$ supported on $I_1, I_1', I_2, I_2'$ that are all odd or all even under $S$ and have norm of at most $1$,
\begin{align}
|\<O_1 O_2\>_\Omega| &\geq |\<O_1 O_1'\>_\Omega| \cdot |\<O_2 O_2'\>_\Omega| - f(\ell_{\text{min}}) \label{o1o2ineq0}\\
|\<O_1 O_2 S_{J \cup I_2}\>_\Omega| &\geq |\<O_1 O_1' S_{I_{1}^+}\>_\Omega| \cdot |\<O_2 O_2' S_{I_{2}^+}\>_\Omega| - f(\ell_{\text{min}}) \label{o1o2sineq0}
\end{align}
where $f(\ell) = \text{poly}(\ell,\epsilon^{-1}) e^{-c \epsilon \ell}$ for some $c > 0$, and $J$ denotes the interval between $I_1$ and $I_2$ (going clockwise).
\end{lemma}

To apply Lemma \ref{ordparam}, consider any $i,j \in X_o$ with $\mathrm{dist}(i,j) > 2 \ell$. Then, it is easy to see that
\begin{align*}
\min(|I_i|, |I_j|, \text{dist}(I_i, I_j)) = \ell
\end{align*}
Therefore the inequality (\ref{o1o2ineq0}), with $\{I_1, I_2\} = \{I_i, I_j\}$ and $\{O_1, O_1', O_2, O_2'\} = \{O_i^\dagger, (O_i')^\dagger, O_j, O_j'\}$ gives
\begin{align*}
|\<O_i^\dagger O_j\>_\Omega| &\geq |\<O_i^\dagger (O_i')^\dagger\>_\Omega| \cdot |\<O_j O_j'\>_\Omega| - f(\ell)
\end{align*}
Using $\<O_i^\dagger (O_i')^\dagger\>_\Omega = \<O_i O_i'\>^*_\Omega$, we deduce that
\begin{align*}
|\<O_i^\dagger O_j\>_\Omega| &\geq |\<O_i O_i'\>_\Omega| \cdot |\<O_j O_j'\>_\Omega| - f(\ell) \nonumber \\
&\geq (1/6)^2 - f(\ell) 
\end{align*}
By the same reasoning, the inequality (\ref{o1o2sineq0}) implies that
\begin{align*}
|\<O_i^\dagger O_j \prod_{k=i+1}^j S_k\>_\Omega| \geq (1/6)^2 - f(\ell)  
\end{align*}
for every $i,j \in X_{d1}$ or $i,j \in X_{d2}$ with $\mathrm{dist}(i,j) > 2 \ell$.

The last step is to note that above expression for $f(\ell)$ implies that there exists $\xi = \tilde{\mathcal{O}}(\frac{1}{\epsilon})$, such that $f(\ell) \leq \frac{1}{72}$ for $\ell \geq \xi$. Therefore, if we choose $\ell = \max(\lambda, \xi)$, then we have 
\begin{align*}
|\<O_i^\dagger O_j\>_\Omega| &\geq \frac{1}{72}  \ \ \ \text{ for all } i,j \in X_o \text{ with }\mathrm{dist}(i,j) > 2 \ell \nonumber \\
|\<O_i^\dagger O_j \prod_{k=i+1}^j S_k\>_\Omega| &\geq \frac{1}{72} \ \ \ \text{ for all } i,j \in X_{d1} \text{ or } i,j \in X_{d2} \text{ with }\mathrm{dist}(i,j) > 2 \ell 
\end{align*}
At this point we have constructed a $(\delta, \ell)$ order parameter on $X_o$, and a $(\delta, \ell)$ disorder parameter on $X_{d1}$ and $X_{d2}$ with $\delta = \frac{1}{72}$ and $\ell = \max(\lambda, \xi) \leq \tilde{\mathcal{O}}(\frac{(\log d)^3}{\epsilon^2})$. This establishes Theorem \ref{mainthm2}.

The proof of Theorem \ref{mainthm1} is almost exactly the same: the only change in this case is that translational symmetry guarantees that $X_o$ and $X_{d1}$ and $X_{d2}$ must either be empty sets or the full set $\{1,...,L\}$. This means that at least one of $X_o$, $X_{d1}$ or $X_{d2} = \{1,...,L\}$, so $|\Omega\>$ has either a $(\delta, \ell)$ order parameter or $(\delta, \ell)$ disorder parameter defined on the whole system, as claimed.
\end{proof}

\subsection{The Fuchs-van de Graaf inequality and the proof of Lemma \ref{orddis}}
In this section, we prove Lemma \ref{orddis}. At the heart of the proof is the following general inequality from quantum information theory. Consider any two spin states $|\psi\>, |\psi'\>$ and any collection of spins, $X \subset \{1,...,L\}$. Then:
\begin{align}
\max_{\text{supp}(A) \subset X} \frac{1}{2} |\<A\>_\psi - \<A\>_{\psi'}| + \max_{\text{supp}(U) \subset X^c} |\<\psi|U|\psi'\>| \geq 1
\label{mainineq}
\end{align}
Here the first maximum is over all operators $A$ that have $\|A\| \leq 1$ and are supported in $X$, and the second maximum is over all unitary operators $U$ that are supported in $X^c$. Intuitively, the first term in Eq. (\ref{mainineq}) measures how distinguishable $|\psi\>, |\psi'\>$ are with respect to observables in $X$, while the second term measures to what extent $|\psi'\>$ can be transformed into $|\psi\>$ by a unitary operation in $X^c$. Thus, (\ref{mainineq}) tells us that if two states look similar with respect to measurements in $X$, then it must be possible to (approximately) transform one state into the other with an operation in $X^c$. 

The inequality (\ref{mainineq}) is easiest to prove in the limiting case where the first term is exactly zero: i.e. $\<A\>_\psi = \<A\>_{\psi'}$ for every $A$ supported in $X$. In this case, $|\psi\>$ and $|\psi'\>$ have identical reduced density matrices in the region $X$. It then follows that there exists a unitary operator $U$, supported in $X^c$, such that $U|\psi'\> = |\psi\>$: indeed, the operator $U$ that does the job is the unitary transformation that rotates the Schmidt states for $|\psi'\>$ in the region $X^c$ into the corresponding Schmidt states for $|\psi\>$. 

To derive the inequality (\ref{mainineq}) in the general case, let $\rho, \rho'$ be the reduced density matrices of $|\psi\>, |\psi'\>$ in the region $X$. The first term in Eq. (\ref{mainineq}) is then equal to the trace distance between $\rho, \rho'$, usually denoted by $D(\rho,\rho')$. Likewise, the second term in Eq. (\ref{mainineq}) can be identified with the square root of the fidelity, $\sqrt{F(\rho,\rho')}$, using Uhlmann's theorem\cite{uhlmann}. With these identifications, Eq. (\ref{mainineq}) is an immediate consequence of the Fuchs-van de Graaf inequality\cite{Fuchs}: $D(\rho,\rho') + \sqrt{F(\rho,\rho')} \geq 1$.

With the inequality (\ref{mainineq}) in hand, we are now ready to prove Lemma \ref{orddis}:
\begin{replemma}{orddis}
Let $|\psi\>$ be a state that is even or odd under $S$. For every pair of disjoint, non-adjacent intervals $I_1, I_2$, the state $|\psi\>$ is either $\delta$ weakly-ordered on $I_1, I_2$ or $(1-\delta)/2$ weakly-disordered on the complementary intervals $J_1, J_2$, where the intervals are ordered $(J_1, I_1, J_2, I_2)$ going clockwise.
\end{replemma}

\begin{proof}
Let $|\psi'\> = S_{I_1} |\psi\>$, and $X = I_1 \cup I_2$. Substituting this into (\ref{mainineq}) gives
\begin{align}
\max_{\text{supp}(A) \subset I_1 \cup I_2} |\<\frac{1}{2}(A- S_{I_1} A S_{I_1})\>_\psi| + \max_{\text{supp}(U) \subset J_1 \cup J_2} |\<U S_{I_1}\>_\psi| \geq 1
\label{mainineq2}
\end{align}
An obvious consequence of Eq. (\ref{mainineq2}) is that either: 
\begin{enumerate}[(i)]
\item{The first term is larger than $\delta$.}
\item{The second term is larger than $1-\delta$.}
\end{enumerate}
We claim that in case (i), $|\psi\>$ is $\delta$ weakly ordered on $I_1, I_2$, while in case (ii), $|\psi\>$ is $(1-\delta)/2$ weakly disordered on $J_1, J_2$. Indeed, in case (i), define
\begin{align*}
B \equiv \frac{1}{4}(A_* - S_{I_1} A_* S_{I_1}) + \frac{1}{4} S ( A_*  - S_{I_1} A_* S_{I_1}) S 
\end{align*}
where $A_*$ is the maximal choice of $A$. By assumption, $|\<B\>_\psi| \geq \delta$. Furthermore, it is easy to see that $B$ is supported on $I_1 \cup I_2$, is odd under $S_{I_1}, S_{I_2}$ and has norm of at most $1$. It follows that $|\psi\>$ is $\delta$ weakly ordered on $I_1, I_2$. In case (ii), define 
\begin{align*}
U_{\pm} \equiv \frac{1}{4}(U_* + S U_* S) \pm \frac{1}{4} S_{J_1} (U_* + S U_* S) S_{J_1}
\end{align*}
where $U_*$ is the maximal choice of $U$. By assumption, $|\<(U_+ + U_-) S_{I_1}\>_\psi| \geq 1-\delta$. Hence either $|\<U_{+} S_{I_1}\>_\psi| \geq (1-\delta)/2$ or $|\<U_{-} S_{I_1}\>_\psi| \geq (1-\delta)/2$. Furthermore, it is easy to see that $U_\pm$ are supported on $J_1 \cup J_2$, are even/odd under $S_{J_1}, S_{J_2}$, and have norm of at most $1$. It follows that $|\psi\>$ is $(1-\delta)/2$ weakly disordered on $J_1, J_2$. 
\end{proof}

\subsection{Some bounds on correlations and the proof of Lemma \ref{ordparam}}
We now give the proof of Lemma \ref{ordparam}. One of the key tools in our proof is a general result about the factorizability of gapped ground states. This result was first proven by Hastings\cite{Hastings_area_law} and subsequently generalized in Ref. \cite{Hamzaetal}. Here we adapt this result to the case of interest, namely Ising symmetric Hamiltonians in a ring geometry:
\begin{theorem}
\label{hastthm}
{\bf (Hastings)}
Let $|\Omega\>$ be the lowest energy even eigenstate of an Ising symmetric Hamiltonian $H$ (\ref{ham}). Then for every interval $I$ and every $\ell > 0$ there exists a corresponding Hermitian projection operator $P_I(\ell)$ that is supported on $I$, is even under $S$, and has the following property: given a partition of $\{1,...,L\}$ into four intervals $I_1, I_2, I_3, I_4$, there exists an operator $P_{bd}(\ell)$ with $\|P_{bd}(\ell)\| \leq 1$, that is supported in the region 
$X_{bd} = \{i : \text{dist}(i, I_n) < \ell/2, \text{for at least two different $n$'s}\}$ and that is even under $S_{X_{bd}^{(k)}}$ for every connected component $X_{bd}^{(k)}$ of $X_{bd}$, such that
\begin{align}
\|  P_{I_1}(\ell) P_{I_2}(\ell) P_{I_3}(\ell) P_{I_4}(\ell) P_{bd}(\ell) P_s - |\Omega\>\<\Omega| \  \| \leq g(\ell) \equiv \text{poly}(\ell,\epsilon^{-1}) e^{-c_1 \epsilon \ell}
\label{Pineq}
\end{align}
where $c_1 > 0$ is a numerical constant, $\epsilon$ is the energy gap within the even subspace, and $P_s$ is the projection onto the even subspace.
\end{theorem}
\begin{proof}
There are only three small differences between Theorem \ref{hastthm} and Hastings' original result (Lemma 1 from Ref. \cite{Hastings_area_law}): First, Theorem \ref{hastthm} applies to Ising symmetric Hamiltonians that have a gap within the \emph{even} subspace while Lemma 1 from Ref. \cite{Hastings_area_law} applies to generic Hamiltonians that have a gap within the \emph{full} Hilbert space. Second, Theorem \ref{hastthm} applies to spin chains with a \emph{ring} geometry while Lemma 1 from Ref. \cite{Hastings_area_law} applies to spin chains with open boundary conditions. Finally, Theorem \ref{hastthm} applies to a partition into \emph{four} intervals while Lemma 1 from Ref. \cite{Hastings_area_law} applies to a partition into \emph{two} intervals. The first difference has only a minor effect on the proof --- all arguments go through exactly as in Ref. \cite{Hastings_area_law}, with the only modification being that one needs to multiply the right hand side of Eq. (A.6) in Ref. \cite{Hastings_area_law} by the projector $P_s$; this $P_s$ is carried through to all subsequent equations, including the final result (A.11). The claim that $P_{bd}(\ell)$ can be chosen to be even under $S_{X_{bd}^{(k)}}$ for every connected component $X_{bd}^{(k)}$ of $X_{bd}$ follows from the fact that the $M_B$ operator appearing in Eq.~(A.10) can be written as a sum, $M_B = \sum_k M_B^{(k)}$, where $M_B^{(k)}$ is even under $S$ and is supported on $X_{bd}^{(k)}$. The second and third differences are even more trivial: it is not hard to see that the proof given in Ref. \cite{Hastings_area_law} works equally well for a ring geometry and for a partition into any finite number of intervals. 
\end{proof}

We now use the above theorem to prove three propositions. The first proposition essentially says that the $P_I$ operators defined above can be used as (approximate) ground state projection operators in appropriate circumstances:
\begin{prop}
\label{tildeO}
{\bf (Gd.~state projectors)}
Let $A$ be an operator that is even under $S$, has norm of at most $1$ and is supported on $I_1 \cup I_2$ where $I_1, I_2$ are two intervals separated by a distance of at least $\ell$. Then,
\begin{align}
\| P_{\tilde{I}_1}(\ell) P_{\tilde{I}_2}(\ell) A |\Omega\> - |\Omega\>\<\Omega| A |\Omega\> \| &\leq \mathcal{O}(\sqrt{g(\ell)}) \label{o1o21}\\
\| P_{\tilde{I}_1}(\ell) P_{\tilde{I}_2}(\ell) A S_{J} |\Omega\> - |\Omega\>\<\Omega| A S_{J} |\Omega\> \| &\leq \mathcal{O}(\sqrt{g(\ell)}) \label{o1o22}
\end{align}
where $\tilde{I}_j \equiv \{i: \text{dist}(i,I_j) < \ell/2\}$ and $J$ denotes the interval between $I_1$ and $I_2$ (going clockwise).
\end{prop}
\begin{proof} 
We will start with (\ref{o1o21}). Let $\tilde{J}$ and $\tilde{K}$ be the two intervals that make up $(\tilde{I}_1 \cup  \tilde{I}_2)^c$. Then $\tilde{I}_1, \tilde{I}_2, \tilde{J}, \tilde{K}$ define a partition of $\{1,...,N\}$ into four intervals. Therefore Theorem \ref{hastthm} implies that
\begin{align}
\|P_{\tilde{I}_1}(\ell) P_{\tilde{I}_2}(\ell) P_{\tilde{J}}(\ell) P_{\tilde{K}}(\ell) P_{bd}(\ell) P_s - |\Omega\>\<\Omega| \| \leq g(\ell)
\label{lemmaineq0}
\end{align}
Multiplying this equation by $A |\Omega\>$, and dropping the $\ell$'s for notational simplicity, we derive:
\begin{align*}
\|P_{\tilde{I}_1} P_{\tilde{I}_2} P_{\tilde{J}} P_{\tilde{K}} P_{bd} P_s A |\Omega\> - |\Omega\>\<\Omega| A |\Omega\> \| \leq g(\ell)
\end{align*}
Reordering the terms, using the fact that $P_{\tilde{J}}, P_{\tilde{K}}, P_{bd}, P_s$ commute with $A$, gives
\begin{align}
\| P_{\tilde{I}_1} P_{\tilde{I}_2} A P_{\tilde{J}} P_{\tilde{K}} P_{bd}  |\Omega\> - |\Omega\>\<\Omega| A |\Omega\> \| \leq g(\ell)
\label{lemmaineq1}
\end{align} 
At the same time, Eq. \ref{lemmaineq0} implies that
\begin{align}
\| P_{\tilde{J}} |\Omega\> - |\Omega\> \| &\leq 2 \sqrt{g(\ell)} \nonumber \\
\| P_{\tilde{K}} |\Omega\> - |\Omega\> \| &\leq 2 \sqrt{g(\ell)} \nonumber \\
\| P_{bd} |\Omega\> - |\Omega\> \| &\leq (2 + \sqrt{2}) \sqrt{g(\ell)} 
\label{lemmaineq2}
 \end{align}
Indeed, this follows from the following claim\cite{Hastings_private}:
\begin{claim}
\label{qalemma}
Suppose that $|\<\psi| P O |\psi\> -1 | \leq \alpha$ where $P$ is a projector and $\|O\| \leq 1$. Then:
\begin{align}
\|P |\psi\> -|\psi\> \| &\leq 2 \sqrt{\alpha} \label{qineq} \\
\|O |\psi\> - |\psi\> \| &\leq (2+\sqrt{2}) \sqrt{\alpha} \label{aineq}
\end{align}
\end{claim}
For example, to derive the first equation in (\ref{lemmaineq2}), set $P = P_{\tilde{J}}$ and $O = P_{\tilde{I}_1}(\ell) P_{\tilde{I}_2}(\ell) P_{\tilde{K}}(\ell) P_{bd}(\ell) P_s$, and $|\psi\> = |\Omega\>$. The other equations follow in a similar manner.  
\begin{proof} (of Claim \ref{qalemma})
Note that the Cauchy-Schwarz inequality implies
\begin{align*}
\sqrt{\<\psi| P |\psi\>} \sqrt{\<\psi| O^\dagger O |\psi\>} \geq |\<\psi| PO|\psi\>| \geq 1 - \alpha
\end{align*}
It follows that $\<\psi| P |\psi\> \geq (1-\alpha)^2$. At the same time, the following inequality holds for any vectors $|\phi\>, |\phi'\>$ with norm at most $1$:
\begin{align}
\| |\phi\> - |\phi'\> \| \leq \sqrt{2 - 2 \text{Re} \<\phi| \phi'\>}
\label{phipsi}
\end{align}
Applying this inequality to $|\psi\>, P|\psi\>$, we derive
\begin{align*}
\| P |\psi\> - |\psi\> \| \leq \sqrt{2 - 2(1-\alpha)^2}
\end{align*}
This establishes (\ref{qineq}). Likewise, applying (\ref{phipsi}) to $O |\psi\>, P|\psi\>$, we deduce that
\begin{align*}
\| O |\psi\> - P |\psi\> \| \leq \sqrt{2 - 2(1-\alpha)}
\end{align*}
By the triangle inequality,
\begin{align*}
\| O |\psi\> - |\psi\> \| \leq \sqrt{2 - 2(1-\alpha)^2} + \sqrt{2 - 2(1-\alpha)}
\end{align*}
This establishes (\ref{aineq}).
\end{proof}
If we now combine (\ref{lemmaineq1}) and (\ref{lemmaineq2}), we deduce that
\begin{align*}
\| P_{\tilde{I}_1} P_{\tilde{I}_2} A|\Omega\> - |\Omega\>\<\Omega|A|\Omega\> \| \leq  (6 + \sqrt{2}) \sqrt{g(\ell)} + g(\ell)
\end{align*}
This completes our derivation of (\ref{o1o21}). As for (\ref{o1o22}), note that Theorem~\ref{hastthm} implies that
\begin{align*}
\|P_{\tilde{I}_1} P_{\tilde{I}_2} P_{\tilde{J}} P_{\tilde{K}} P_{bd} P_s A S_{J} |\Omega\> - |\Omega\>\<\Omega| A S_{J} |\Omega\> \| \leq g(\ell)
\end{align*}
Similarly to before, reordering the terms gives
\begin{align}
\| P_{\tilde{I}_1} P_{\tilde{I}_2} A S_{J} P_{\tilde{J}} P_{\tilde{K}} P_{bd} |\Omega\> - |\Omega\>\<\Omega| A S_{J}|\Omega\> \| \leq g(\ell)
\label{lemmaineq3}
\end{align} 
By the same reasoning as above, Eq. \ref{lemmaineq2} implies the desired inequality:
\begin{align*}
\| P_{\tilde{I}_1} P_{\tilde{I}_2} A S_{J} |\Omega\> - |\Omega\>\<\Omega|A S_{J}|\Omega\> \| \leq  (6 + \sqrt{2}) \sqrt{g(\ell)} + g(\ell)
\end{align*}
\end{proof}

We now move on to the second proposition --- a variation of Hastings' result that gapped ground states have short range correlations\cite{Hastings_LSM}:
\begin{prop}
\label{src}
{\bf (Short-range correlations)} Let $O_1, O_2$ be two operators that are even under the symmetry, have norm of at most $1$ and are supported on intervals $I_1, I_2$ separated by a distance of at least $\ell$. Then,
\begin{align}
| \<\Omega|O_1 O_2|\Omega\> - \<\Omega|O_1|\Omega\>\<\Omega|O_2|\Omega\>| \leq \mathcal{O}(\sqrt{g(\ell)})  
\label{srcineqlemma}
\end{align}
\end{prop}

\begin{proof}
By construction, $I_1, I_1^c$ define a partition of $\{1,...,L\}$ into two intervals. Therefore, according to Theorem~\ref{hastthm} (with $I_3$ and $I_4$ being empty sets), 
\begin{align}
\| P_{I_1^c}(\ell) P_{I_1}(\ell) P_{bd}(\ell) P_s - |\Omega\>\<\Omega| \  \| \leq g(\ell) 
\label{Pineq2}
\end{align}
Sandwiching this inequality between $\<\Omega |O_1$ and $O_2 |\Omega\>$ gives 
\begin{align*}
|\<\Omega| O_1 P_{I_1^c} P_{I_1} P_{bd} P_s O_2 |\Omega\> -
\<\Omega| O_1|\Omega\>\<\Omega| O_2 |\Omega\>| \leq g(\ell)
\end{align*}
where we have dropped the $\ell$'s for clarity. Reordering the terms, we obtain
\begin{align}
|\<\Omega| P_{I_1^c} O_1 O_2  P_{I_1} P_{bd} |\Omega\> -
\<\Omega| O_1|\Omega\>\<\Omega| O_2 |\Omega\>| \leq g(\ell)
\label{srcprop1}
\end{align}
where we have used the fact that $O_1$ commutes with $P_{I_1^c}$ and $O_2$ commutes with $P_s, P_{bd}, P_{I_1}$.

At the same time, using (\ref{Pineq2}) and Claim \ref{qalemma}, we deduce
\begin{align}
\| P_{I_1} |\Omega\> - |\Omega\> \| &\leq 2 \sqrt{g(\ell)}  \nonumber \\
\| P_{I_1^c} |\Omega\> - |\Omega\> \| &\leq 2 \sqrt{g(\ell)} \nonumber \\
\| P_{bd} |\Omega\> - |\Omega\> \| &\leq (2 + \sqrt{2}) \sqrt{g(\ell)} 
\label{Pomegaineq}
 \end{align}
Combining (\ref{srcprop1}) and (\ref{Pomegaineq}) gives:
\begin{align*}
|\<\Omega| O_1 O_2 |\Omega\> - \<\Omega| O_1 |\Omega\>\<\Omega|O_2 |\Omega\>| \leq (6 + \sqrt{2}) \sqrt{g(\ell)} + g(\ell)
\end{align*}
This implies the desired inequality (\ref{srcineqlemma}).
\end{proof}

Our final proposition is also a result about short-range correlations, but in the more complicated case where the operators in question are either \emph{odd} under symmetry or involve a symmetry transformation along an interval:
\begin{prop}
\label{factor}
{\bf (Generalized short-range corr.)}
For any intervals $I_1, I_2$ separated by a distance of at least $\ell$, there exists three pairs of operators $(\rho_1,\rho_2)$, $(\rho^e_1,\rho^e_2)$, $(\rho^o_1,\rho^o_2)$ supported on $I_1, I_2$ with 
\begin{align}
|\<\Omega| O_1 O_2 |\Omega\> - \text{Tr}(O_1 \rho_1) \text{Tr}(O_2 \rho_2)| &\leq h(\ell) \equiv \text{poly}(\ell,\epsilon^{-1}) e^{-c_2 \epsilon \ell} \label{factor1} \\
|\<\Omega| O_1 O_2 S_{J}|\Omega\> - \text{Tr}(O_1 \rho^e_1) \text{Tr}(O_2 \rho^e_2)| &\leq h(\ell) \label{factor2} \\
|\<\Omega| O_1 O_2 S_{J}|\Omega\> - \text{Tr}(O_1 \rho^o_1) \text{Tr}(O_2 \rho^o_2)| &\leq h(\ell) \label{factor3}
\end{align}
where Eqs. (\ref{factor1}) and (\ref{factor3}) hold for every $O_1, O_2$ that are supported on $I_1, I_2$, are odd under $S$, and have norm of at most $1$, and Eq. (\ref{factor2}) holds for every $O_1, O_2$ that are supported on $I_1, I_2$, are even under $S$, and have norm of at most $1$. Here, $c_2 > 0$ is a constant, and $J$ is the interval between $I_1, I_2$ (going clockwise).
\end{prop}
\begin{proof}
We start with Eq. (\ref{factor1}). To derive this result, define $\tilde{I}_j = \{i : \text{dist}(i,I_j) < \ell/4\}$ and let
\begin{align*}
P_1 \equiv P_{\tilde{I}_1}(\ell/2), \quad P_2 \equiv P_{\tilde{I}_2}(\ell/2)
\end{align*}
Also, let $U_1, U_2$ be any two operators that are supported on $I_1, I_2$, are odd under $S$ and have norm of at most $1$. We claim that the following inequalities hold:
\begin{align*}
|\<\Omega| O_1 O_2 |\Omega\>\<\Omega| U_1 U_2 |\Omega\> -
\<\Omega| O_1 O_2 P_1 P_2 U_1  U_2 |\Omega\>| &\leq \mathcal{O}(\sqrt{g(\ell/2)}) \nonumber \\
|\<\Omega| O_1 O_2 P_1 P_2 U_1 U_2 |\Omega\> -
\<\Omega| O_1 P_1 U_1 |\Omega\>\<\Omega| O_2 P_2 U_2 |\Omega\>| &\leq \mathcal{O}(\sqrt{g(\ell/2)}) 
\end{align*}
Here the first inequality follows from Proposition \ref{tildeO}, while the second inequality follows from the fact that $|\Omega\>$ has short-range correlations (Proposition \ref{src}) and that $O_1 P_1 U_1$ and $O_2 P_2 U_2$ are supported on two intervals that are separated by a distance of at least $\ell/2$. 

Adding together these two inequalities and using the triangle inequality, we derive
\begin{align}
|\<\Omega| O_1 O_2 |\Omega\>\<\Omega| U_1 U_2 |\Omega\> - 
\<\Omega| O_1 P_1 U_1  |\Omega\>\<\Omega| O_2 P_2 U_2 |\Omega\>| \leq \mathcal{O}(\sqrt{g(\ell/2)})
\label{factorineq}
\end{align}
Now, let us choose $U_1, U_2$ such that $|\<\Omega | U_1 U_2 | \Omega\>|$ is as large as possible. There are two cases to consider: $|\<\Omega | U_1 U_2 | \Omega\>| \leq \mathcal{O}(\sqrt[4]{g(\ell/2)})$ or
$|\<\Omega | U_1 U_2 | \Omega\>| > \mathcal{O}(\sqrt[4]{g(\ell/2)})$. In the first case, Eq. (\ref{factor1}) holds with $\rho_1 = \rho_2 = 0$ and
$ h(\ell) = \mathcal{O}(\sqrt[4]{g(\ell/2)})$. In the second case, we can divide (\ref{factorineq})
by $|\<\Omega | U_1 U_2 | \Omega\>|$ to obtain
\begin{align*}
|\<\Omega| O_1 O_2 |\Omega\> - \text{Tr}(O_1 \rho_1) \text{Tr}(O_2 \rho_2)| &\leq \mathcal{O}(\sqrt[4]{g(\ell/2)})
\end{align*}
with
\begin{align*}
\rho_1 = \frac{\text{Tr}_{I_1^c}(P_1 U_1 |\Omega\>\<\Omega|)}{\sqrt{|\<\Omega| U_1 U_2 |\Omega\>|}}, \quad \rho_2 = \frac{\text{Tr}_{I_2^c}(P_2 U_2 |\Omega\>\<\Omega|)}{\sqrt{|\<\Omega| U_1 U_2 |\Omega\>|}}.
\end{align*}
Again this establishes the desired inequality, Eq. (\ref{factor1}), with
$h(\ell) = \mathcal{O}(\sqrt[4]{g(\ell/2)})$. 

The proof of Eq. (\ref{factor2}) is similar. In this case, let $U_1, U_2$ be any two operators that are supported on $I_1, I_2$, are \emph{even} under $S$ and have norm of at most $1$. We claim that the following inequalities hold:
\begin{align*}
|\<\Omega| O_1 O_2 S_{J}|\Omega\>\<\Omega| U_1 U_2 S_{J}|\Omega\> -
\<\Omega| O_1 O_2 S_{J} P_1 P_2 U_1  U_2 S_{J}|\Omega\>| &\leq O(\sqrt{g(\ell/2)}) \nonumber \\
|\<\Omega| O_1 O_2 S_{J} P_1 P_2 U_1 U_2 S_{J}|\Omega\> -
\<\Omega| O_1 S_{J} P_1 U_1 S_{J}|\Omega\>\<\Omega| O_2 S_{J} P_2  U_2 S_{J}|\Omega\>| &\leq \mathcal{O}(\sqrt{g(\ell/2)}) 
\end{align*}
As above, the first inequality follows from Proposition \ref{tildeO}, while the second inequality follows from the fact that $|\Omega\>$ has short-range correlations (Proposition \ref{src}). 

Adding together these three inequalities, we derive
\begin{align}
|\<\Omega| O_1 O_2 S_{J}|\Omega\>\<\Omega| U_1 U_2 S_{J}|\Omega\>  - 
\<\Omega| O_1 S_{J} P_1 U_1 S_{J}|\Omega\>\<\Omega| O_2 S_{J} P_2  U_2 S_{J}|\Omega\> \leq \mathcal{O}(\sqrt{g(\ell/2)})
\label{factorineq2}
\end{align}
As above, we now choose $U_1, U_2$ such that $|\<\Omega | U_1 U_2 S_{J} | \Omega\>|$ is as large as possible. If $|\<\Omega | U_1 U_2 S_J | \Omega\>| \leq \mathcal{O}(\sqrt[4]{g(\ell/2)})$, then
Eq. (\ref{factor2}) holds with $\rho^e_1 = \rho^e_2 = 0$. On the other hand, if $|\<\Omega | U_1 U_2 S_J | \Omega\>| > \mathcal{O}(\sqrt[4]{g(\ell/2)})$, then we can divide (\ref{factorineq2})
by $|\<\Omega | U_1 U_2 S_J | \Omega\>|$ to obtain 
\begin{align*}
|\<\Omega| O_1 O_2 S_{J}|\Omega\> - \text{Tr}(O_1 \rho^e_1) \text{Tr}(O_2 \rho^e_2)| &\leq \mathcal{O}(\sqrt[4]{g(\ell/2)})
\end{align*}
with
\begin{align*}
\rho^e_1 = \frac{\text{Tr}_{I_1^c}(S_{J} P_1 U_1 S_{J}|\Omega\>\<\Omega|)}{\sqrt{|\<\Omega| U_1 U_2 S_{J}|\Omega\>|}}, \quad \rho^e_2 = \frac{\text{Tr}_{I_2^c}(S_{J}P_2 U_2 S_{J}|\Omega\>\<\Omega|)}{\sqrt{|\<\Omega| U_1 U_2 S_{J}|\Omega\>|}}.
\end{align*}
This gives the desired inequality, Eq. (\ref{factor2}), with $h(\ell) = \mathcal{O}(\sqrt[4]{g(\ell/2)})$. The proof of Eq. (\ref{factor3}) follows in exactly the same way.

\end{proof}

We are now ready to prove Lemma \ref{ordparam}:
\begin{replemma}{ordparam}
Let $I_1, I_2$ be two disjoint intervals and let $\ell_{\text{min}} = \text{min}(|I_1|, |I_2|, \text{dist}(I_1, I_2))$. For any operators $O_1, O_1', O_2, O_2'$ supported on $I_1, I_1', I_2, I_2'$ that are all odd or all even under $S$ and have norm of at most $1$,
\begin{align}
|\<O_1 O_2\>_\Omega| &\geq |\<O_1 O_1'\>_\Omega| \cdot |\<O_2 O_2'\>_\Omega| - f(\ell_{\text{min}}) \label{ordineq}\\
|\<O_1 O_2 S_{J \cup I_2}\>_\Omega| &\geq |\<O_1 O_1' S_{I_{1}^+}\>_\Omega| \cdot |\<O_2 O_2' S_{I_{2}^+}\>_\Omega| - f(\ell_{\text{min}})
\label{disordineq}
\end{align}
where $f(\ell) = \text{poly}(\ell,\epsilon^{-1}) e^{-c \epsilon \ell}$ for some $c > 0$, and $J$ denotes the interval between $I_1$ and $I_2$ (going clockwise).
\end{replemma}

\begin{proof}
We begin with (\ref{ordineq}). Let $\tilde{I} = \{i : \text{dist}(i,I) < \ell_{\text{min}}/4\}$ and 
\begin{align*}
P_1 \equiv P_{\tilde{I}_1}(\ell_{\text{min}}/2), \quad P_2 \equiv P_{\tilde{I}_2}(\ell_{\text{min}}/2), \quad P_1' \equiv P_{\tilde{I}_1'}(\ell_{\text{min}}/2),\quad P_2' \equiv P_{\tilde{I}_2'}(\ell_{\text{min}}/2)
\end{align*} 
The proof proceeds in three steps. First, we note that $P_1, O_1$ and $P_2, O_2$ are supported on two intervals $\tilde{I}_1, \tilde{I}_2$ which are separated by a distance of at least $\ell_{\text{min}}/2$. Therefore Proposition \ref{src} implies
\begin{align*}
|\<\Omega| (O_1^\dagger P_1 O_1) (O_2^\dagger P_2 O_2) |\Omega\> - 
\<\Omega| O_1^\dagger P_1 O_1 |\Omega\> \<\Omega|  O_2^\dagger P_2 O_2 |\Omega\>| \leq \mathcal{O}(\sqrt{g(\ell_{\text{min}}/2)})
\end{align*}
It follows that
\begin{align}
\| P_1 P_2 O_1 O_2 |\Omega\> \|^2
\geq \| P_1 O_1 |\Omega\>\|^2 \cdot \| P_2 O_2 |\Omega\>\|^2 - \mathcal{O}(\sqrt{g(\ell_{\text{min}}/2)})
\label{srcineq0}
\end{align}
Next observe that $\| P_1 O_1 |\Omega\> \| \geq \| P_1' O_1' P_1 O_1 |\Omega\>\|$ and $\| P_2 O_2 |\Omega\> \| \geq \| P_2' O_2' P_2 O_2 |\Omega\>\|$ since \\
$\|O_1'\|, \|O_2'\| \leq 1$. Substituting these inequalities into (\ref{srcineq0}), we conclude that 
\begin{align}
\| P_1 P_2 O_1 O_2 |\Omega\> \|
\geq \| P_1 P_1' O_1 O_1' |\Omega\>\| \cdot \| P_2 P_2' O_2 O_2' |\Omega\>\| - \mathcal{O}(\sqrt{g(\ell_{\text{min}}/2)})
\label{srcineq}
\end{align}
Next we invoke Proposition \ref{tildeO} which implies that
\begin{align}
\|P_1 P_2 O_1 O_2 |\Omega\> \| &\leq |\<\Omega| O_1 O_2 |\Omega\>| + \mathcal{O}(\sqrt{g(\ell_{\text{min}}/2)}) \nonumber \\
\| P_1 P_1' O_1 O_1' |\Omega\>\| &\geq |\<\Omega|O_1 O_1' |\Omega\>| - \mathcal{O}(\sqrt{g(\ell_{\text{min}}/2)}) \nonumber \\
\| P_2 P_2' O_2 O_2' |\Omega\>\| &\geq |\<\Omega|O_2 O_2' |\Omega\>| - \mathcal{O}(\sqrt{g(\ell_{\text{min}}/2)}) 
\label{12ineq}
\end{align}

Combining (\ref{srcineq}),(\ref{12ineq}), we derive:
\begin{align*}
|\<\Omega| O_1 O_2 |\Omega\>| \geq |\<\Omega|O_1 O_1' |\Omega\>| \cdot |\<\Omega|O_2 O_2' |\Omega\>| - \mathcal{O}(\sqrt{g(\ell_{\text{min}}/2)})
\end{align*}
This establishes the desired inequality (\ref{ordineq}) with $f(\ell) = \mathcal{O}(\sqrt{g(\ell/2)})$. 

The derivation of (\ref{disordineq}) is similar: first we note that
\begin{align*}
\| P_1 P_2 O_1 O_2 S_{J \cup I_2}|\Omega\> \|^2
\geq \| P_1 O_1 S_{J \cup I_2} |\Omega\>\|^2 \cdot \| P_2 O_2 S_{J \cup I_2} |\Omega\>\|^2 - \mathcal{O}(\sqrt{g(\ell_{\text{min}}/2)})
\end{align*}
by the same reasoning as in (\ref{srcineq0}). Next we observe that
\begin{align*}
\|P_1 O_1 S_{J \cup I_2} |\Omega\>\|^2 = \| P_1 O_1 S_{I_1^+} |\Omega\>\|^2, \quad \| P_2 O_2 S_{J \cup I_2} |\Omega\>\|^2 =  \| P_2 O_2 S_{I_2^+} |\Omega\>\|^2
\end{align*}
where the first equality follows from $S_{J \cup I_2} (O_1^\dagger P_1 O_1)  S_{J \cup I_2} = S_{I_1^+} (O_1^\dagger P_1 O_1) S_{I_1^+}$ and similarly for the second equality. Also, we have the inequalities 
\begin{align*}
\| P_1 O_1 S_{I_1^+}|\Omega\> \| \geq \| P_1' O_1' P_1 O_1 S_{I_1^+}|\Omega\>\|, \quad \quad
\| P_2 O_2 S_{I_2^+}|\Omega\> \| \geq \| P_2' O_2' P_2 O_2 S_{I_2^+}|\Omega\>\|
\end{align*}
Putting this together gives:
\begin{align}
\| P_1 P_2 O_1 O_2 S_{J \cup I_2}|\Omega\> \|
\geq \| P_1 P_1' O_1 O_1' S_{I_1^+}|\Omega\>\| \cdot \| P_2 P_2' O_2 O_2' S_{I_2^+}|\Omega\>\| - \mathcal{O}(\sqrt{g(\ell_{\text{min}}/2)})
\label{srcineq2}
\end{align}
Next, we invoke Proposition \ref{tildeO} to deduce 
\begin{align}
\| P_1 P_2 (O_1 O_2 S_{I_2}) S_{J}|\Omega\> \| &\leq |\<\Omega| (O_1 O_2 S_{I_2}) S_{J}|\Omega\>| + \mathcal{O}(\sqrt{g(\ell_{\text{min}}/2)}) \nonumber \\
\| P_1 P_1' O_1 O_1' S_{I_1^+}|\Omega\>\| &\geq |\<\Omega|O_1 O_1' S_{I_1^+}|\Omega\>| - \mathcal{O}(\sqrt{g(\ell_{\text{min}}/2)}) \nonumber \\
\| P_2 P_2' O_2 O_2' S_{I_2^+}|\Omega\>\| &\geq |\<\Omega| O_2 O_2' S_{I_2^+}|\Omega\>| - \mathcal{O}(\sqrt{g(\ell_{\text{min}}/2)})
\label{12ineq2}
\end{align}

Combining (\ref{srcineq2}), (\ref{12ineq2}), we derive
\begin{align*}
|\<\Omega| O_1 O_2 S_{J \cup I_2}|\Omega\>| \geq |\<\Omega|O_1 O_1' S_{I_1^+}|\Omega\>| \cdot |\<\Omega|O_2 O_2' S_{I_2^+}|\Omega\>| - \mathcal{O}(\sqrt{g(\ell_{\text{min}}/2)})
\end{align*}
This establishes the desired inequality (\ref{disordineq}) with
$f(\ell) = \mathcal{O}(\sqrt{g(\ell/2)})$. 
\end{proof}

\subsection{A bound on entanglement and the proof of Lemma \ref{weakimpstrong}}
In this section we prove Lemma \ref{weakimpstrong}. The key tool in our proof is a bound on entanglement in gapped ground states due to Arad, Kitaev, Landau, and Vazirani
\cite{AKLV}. Here we adapt this bound to the case of Ising symmetric Hamiltonians:
\begin{theorem}
\label{schmidtthm}
{\bf (AKLV)} 
Let $|\Omega\>$ be the lowest energy even eigenstate of an Ising symmetric Hamiltonian $H$ (\ref{ham}) and let $\epsilon$ be the energy gap in the even subspace. Then for every bipartition of the spin chain into two intervals, and every $\Delta > 0$, there exists an even state $|\psi\>$ with overlap $| \<\psi|\Omega\> | \geq  1-\Delta$ and with a Schmidt rank across the cut of at most $s = e^{\tilde{\mathcal{O}}([\log d]^3/\epsilon)}$.\footnote{Here we treat $\Delta$ as a constant when using $\tilde{\mathcal{O}}$ notation.}
\end{theorem}
\begin{proof}
There are only two small differences between Theorem \ref{schmidtthm} and the original result of Arad, Kitaev, Landau, and Vazirani (Lemma 6.3 from Ref. \cite{AKLV}): first, Theorem \ref{schmidtthm} applies to Ising symmetric Hamiltonians that have a gap within the \emph{even} subspace, while Lemma 6.3 from Ref. \cite{AKLV} applies to generic Hamiltonians with a gap within the full Hilbert space. Second, Theorem \ref{schmidtthm} applies to spin chains in a ring geometry while Lemma 6.3 from Ref. \cite{AKLV} applies to spin chains with open boundary conditions. The first modification does not affect the proof in any way: the approximate ground state projector (AGSP) constructed in Ref. \cite{AKLV} is guaranteed to be Ising symmetric as long as the Hamiltonian is, so the same arguments go through in the Ising symmetric case as in Lemma 6.3 in Ref. \cite{AKLV}. The second modification can also be accommodated trivially since any spin chain in a ring geometry can be rewritten as a spin chain with open boundary conditions, by ``squashing'' the ring --- that is, clustering pairs of spins $i, L-i$ into a single superspin with dimension $d^2$.  
\end{proof}

We now proceed with the proof of Lemma \ref{weakimpstrong}:
\begin{replemma}{weakimpstrong}
For every $\delta > 0$, there exists a length $\lambda = \tilde{\mathcal{O}}(\frac{(\log d)^3}{\epsilon^2})$ such that if $|\Omega\>$ is $\delta$ weakly-ordered on $I_1, I_2$ then it is $\delta/2$ strongly-ordered on $I_1, I_2$ for all $I_1, I_2$ that are separated by a distance of at least $\lambda$.  Likewise, if $|\Omega\>$ is $\delta$ weakly-disordered on $I_1, I_2$ with even (odd) parity then it is $\delta/2$ strongly-disordered on $I_1, I_2$ with even (odd) parity for all $I_1, I_2$ separated by at least $\lambda$.   
\end{replemma}
\begin{proof}
We begin with the first part of the Lemma. Suppose that $|\Omega\>$ is $\delta$ weakly ordered on $I_1, I_2$. We wish to show that it is $\delta/2$ strongly ordered on $I_1, I_2$ as long as $I_1, I_2$ are separated by a sufficiently large distance. The first step is to invoke Theorem \ref{schmidtthm} to construct a
state $|\psi\>$ that has a Schmidt rank $s = e^{\tilde{\mathcal{O}}([\log d]^3/\epsilon)}$ for the bipartition $I_1 \cup I_1^c$ and that has an overlap $| \<\psi|\Omega\> | \geq  1-\Delta$. For the moment, we leave $\Delta$ undetermined --- we will choose a specific value later. 

By definition, we can Schmidt decompose $|\psi\>$ using at most $s$ Schmidt states: 
\begin{align*}
|\psi\> = \sum_{\alpha=1}^{s} \lambda^\alpha |\alpha\> \otimes |\phi^\alpha\>
\end{align*}
Here $\{|\alpha\>, \alpha = 1,...,d^{|I_1|}\}$ are the Schmidt basis states corresponding to the region $I_1$, and $|\phi^\alpha\>$ are the Schmidt states in region $I_1^c$, ordered in such a way that the Schmidt coefficients $\lambda^\alpha$ vanish for $\alpha > s$.

Now since $|\Omega\>$ is $\delta$ weakly ordered on $I_1, I_2$, there exists an operator $A$ supported on $I_1 \cup I_2$ that is odd under $S_{I_1}$ and $S_{I_2}$ and satisfies $\<\Omega| A |\Omega\> \geq \delta$. Like any operator acting on $I_1 \cup I_2$ we can decompose $A$ into a sum
\begin{align}
A = \sum_{\alpha, \beta = 1}^{d^{|I_1|}} O^{\alpha \beta}_1 \otimes O^{\alpha \beta}_2
\label{Aexp}
\end{align}
where $O^{\alpha \beta}_1 \equiv |\alpha\>\<\beta|$ and where the $O^{\alpha \beta}_2$ are some (undetermined) operators acting on $I_2$. Since $\|A\| \leq 1$, we must have $\|O^{\alpha \beta}_2\| \leq 1$. Also, since $A$ is odd under $S_{I_1}$ and $S_{I_2}$, we know that $O^{\alpha \beta}_1, O^{\alpha \beta}_2$ are odd under the symmetry for all the nonvanishing terms in (\ref{Aexp}).
 
Now consider the expectation value $\<\psi|A|\psi\>$. It is clear from the expression for $|\psi\>$ that the only terms in (\ref{Aexp}) that contribute are those with $\alpha, \beta \leq s$, so that
\begin{align*}
\<\psi|A|\psi\> = \<\psi| A_{\text{trun}}|\psi\>, \quad \quad A_{\text{trun}} \equiv \sum_{\alpha, \beta =1}^{s} O^{\alpha \beta}_1 \otimes O^{\alpha \beta}_2
\end{align*}
At the same time, the formula for the trace distance between pure states implies that 
\begin{align*}
|\<\Omega|A |\Omega\> - \<\psi|A |\psi\> | &\leq 2\sqrt{1-| \<\psi|\Omega\> |^2} \leq  2\sqrt{2\Delta} \nonumber \\
|\<\Omega|A_{\text{trun}}|\Omega\> - \<\psi|A_{\text{trun}}|\psi\> | &\leq 2\sqrt{1-| \<\psi|\Omega\> |^2} \leq  2\sqrt{2\Delta}
\end{align*}
(In the second equation we use $\|A_{\text{trun}}\| \leq 1$ which follows from the fact that $A_{\text{trun}} = PAP$ where $P = \sum_{\alpha=1}^s |\alpha\>\<\alpha|$). 
Combining these three equations with the fact that $\<\Omega|A|\Omega\> \geq \delta$, implies the lower bound
\begin{align}
|\<\Omega|A_{\text{trun}}|\Omega\>| \geq \delta - 4\sqrt{2\Delta}
\label{Atrunineq}
\end{align}
Next we invoke Proposition \ref{factor}. In particular, we use Eq. (\ref{factor1}) to deduce that
\begin{align}
|\<\Omega| A_{\text{trun}} |\Omega\> - \text{Tr}(A_{\text{trun}} [\rho_1 \otimes \rho_2])| &\leq s^2 h(\ell)
\label{Atrunfactorineq}
\end{align}
since $A_{\text{trun}}$ is a sum of $s^2$ terms, each of the form $O_1 O_2$.

To proceed further, consider the problem of maximizing the quantity $|\text{Tr}(B [\rho_1 \otimes \rho_2])|$ subject to the constraint that $B$ is supported on $I_1 \cup I_2$, is odd under $S_{I_1}$ and $S_{I_2}$, and has norm at most $1$. Given that $\rho_1 \otimes \rho_2$ is a tensor product of two operators supported on $I_1, I_2$ it is easy to see that the maximal choice of $B$ is a product of the form $B_* = O_{1*} \otimes O_{2*}$ where $O_{1*}, O_{2*}$ are supported on $I_1, I_2$, are odd under the symmetry, and have norm of at most $1$.\footnote{In fact, one can derive an explicit formula for $O_{1*}, O_{2*}$: the operator $O_{i*} = U_i^\dagger$ where $U_i$ is the unitary operator appearing in the polar decomposition of $\rho_i^- = \frac{1}{2}(\rho_i - S\rho_i S)$.} In particular, this means that  
\begin{align}
|\text{Tr}(A_{\text{trun}} [\rho_1 \otimes \rho_2])| \leq  |\text{Tr}(O_{1*}\rho_1) \text{Tr}(O_{2*}\rho_2)|
\label{AtrunOineq}
\end{align}
At the same time, we can apply Eq. (\ref{factor1}) again to derive the inequality
\begin{align}
|\<\Omega| O_{1*} O_{2*} |\Omega\> - \text{Tr}(O_{1*} \rho_1) \text{Tr}(O_{2*} \rho_2)| &\leq h(\ell)
\label{o1o2ineq}
\end{align}

We now combine the four inequalities (\ref{Atrunineq}), (\ref{Atrunfactorineq}), (\ref{AtrunOineq}), (\ref{o1o2ineq}) to conclude that
\begin{align*}
|\<\Omega| O_{1*} O_{2*} |\Omega\>| \geq \delta - 4\sqrt{2\Delta} - (s^2 + 1) h(\ell)
\end{align*}
The last step is to choose $\Delta$ so that $4\sqrt{2\Delta} = \delta/4$. The inequality then becomes
\begin{align*}
|\<\Omega| O_{1*} O_{2*} |\Omega\>| \geq 3\delta/4 - (s^2 + 1)  h(\ell)
\end{align*}
Finally, comparing the expressions for $h(\ell)$ and $s$, we see that there exists $\lambda = \tilde{\mathcal{O}}(\frac{(\log d)^3}{\epsilon^2})$ such that if $\ell \geq \lambda$ then $(s^2 + 1)  h(\ell) \leq \delta/4$ and hence $|\<\Omega| O_{1*} O_{2*} |\Omega\>| \geq \delta/2$. This proves the first part of the Lemma.

To prove the second part, we use the same logic. Suppose that $|\Omega\>$ is $\delta$ weakly disordered with even parity. Then there exists an operator $A$ supported on $I_1 \cup I_2$ that is even under $S_{I_1}$ and satisfies $\<\Omega| A S_{J}|\Omega\> \geq \delta$. As above, we can approximate $|\Omega\>$ by a state $|\psi\>$ with overlap $1-\Delta$ and Schmidt rank $s$ for the bipartition $I_1 \cup I_1^c$. Decomposing $|\psi\>$ and $A$ in the same way as above, we have
\begin{align*}
\<\psi|A S_{J}|\psi\> = \<\psi| A_{\text{trun}} S_{J}|\psi\>, \quad \quad A_{\text{trun}} \equiv \sum_{\alpha, \beta =1}^{s} O^{\alpha \beta}_1 \otimes O^{\alpha \beta}_2  
\end{align*}
where $O^{\alpha \beta}_1, O^{\alpha \beta}_2$ are both even under $S$ and have norm at most $1$. Following the same logic as before gives the lower bound
\begin{align*}
|\<\Omega|A_{\text{trun}} S_{J}|\Omega\>| \geq \delta - 4\sqrt{2\Delta}
\end{align*}
Next we apply the above proposition to derive the inequality
\begin{align*}
|\<\Omega| A_{\text{trun}} S_{J} |\Omega\> - \text{Tr}(A_{\text{trun}} [\rho^e_1 \otimes \rho^e_2])| &\leq s^2 h(\ell) 
\end{align*} 
In the same way as before we have
\begin{align*}
|\text{Tr}(A_{\text{trun}} [\rho^e_1 \otimes \rho^e_2])| &\leq  |\text{Tr}(O_{1*}\rho^e_1) \text{Tr}(O_{2*}\rho^e_2)| 
\end{align*}
for some $O_{1*}, O_{2*}$ that are supported on $I_1, I_2$, are even under the symmetry, and have norm of at most $1$. Also, 
\begin{align*}
|\<\Omega| O_{1*} O_{2*} S_{J}|\Omega\> - \text{Tr}(O_{1*} \rho^e_1) \text{Tr}(O_{2*} \rho^e_2)| &\leq h(\ell) 
\end{align*}
Putting this all together gives 
\begin{align*}
|\<\Omega| O_{1*} O_{2*} S_{J}|\Omega\>| \geq \delta - 4\sqrt{2\Delta} - (s^2 + 1) h(\ell)
\end{align*}
Choosing $\Delta, \lambda$ as before gives the desired inequality $|\<\Omega| O_{1*} O_{2*} S_{J}|\Omega\>| \geq \delta/2$ for $\ell \geq \lambda$. Exactly the same argument works when $|\Omega\>$ is $\delta$ weakly disordered with \emph{odd} parity.
\end{proof}

\section{Proofs of additional constraints}
\label{sec-two-more-theorems}
In this section, we present the proofs of the two additional constraints on order and disorder parameters: Theorems \ref{mutconst} and \ref{exchconst}.

\begin{reptheorem}{mutconst}
If $|\Omega\>$ has a $(\delta, \ell)$ order parameter defined at two points $i_1, i_3$ and a $(\delta, \ell)$ disorder parameter defined at two points $i_2, i_4$ with $i_1 < i_2 < i_3 < i_4$, then
\begin{align*}
\min_{k,l}[\mathrm{dist}(i_k,i_l)] \leq 2\ell + \tilde{\mathcal{O}}\left(\frac{\log \delta^{-1}}{\epsilon} \right)  
\end{align*}
\end{reptheorem}
\begin{proof}
Let $O_1, O_3$ and $O_{2}, O_{4}$ be the postulated order and disorder parameters, respectively. Without loss of generality, we can assume that
\begin{align*} 
\<\Omega| O_{1}^\dagger O_{3} |\Omega\> = \<\Omega| O_{2}^\dagger O_{4} \prod_{k=i_2+1}^{i_4} S_k|\Omega\> = \delta, 
\end{align*}
since we can always make these equalities hold by replacing $O_i \rightarrow \kappa_i O_i$ for appropriate $|\kappa_i| \leq 1$. Let $\ell_{\text{min}} = \min_{k,l}[\mathrm{dist}(i_k,i_l)] - 2\ell$ and define
\begin{align*}
U &= P_1 P_3 O_1^\dagger O_3, \\
V &= P_2 P_4 O_2^\dagger O_4 \prod_{k=i_2+1}^{i_4} S_k,
\end{align*}
where
\begin{align}
P_1 \equiv P_{\tilde{K}_1}(\ell_{\text{min}}), \quad P_2 \equiv P_{\tilde{K}_2}(\ell_{\text{min}}), \quad P_3 \equiv P_{\tilde{K}_3}(\ell_{\text{min}}), \quad P_4 \equiv P_{\tilde{K}_4}(\ell_{\text{min}})
\label{Pidef}
\end{align}
and
where $K_j = [i_j - \ell, i_j + \ell]$ and $\tilde{K}_j = \{i : \text{dist}(i,K_j) < \ell_{\text{min}}/2\}$. Then according to Proposition \ref{tildeO},
\begin{align*}
\|U|\Omega\> - \delta|\Omega\>\| &\leq \mathcal{O}(\sqrt{g(\ell_{\text{min}})}) \nonumber \\
\|V|\Omega\> - \delta|\Omega\>\| &\leq \mathcal{O}(\sqrt{g(\ell_{\text{min}})})
\end{align*}
It follows that
\begin{align}
\| UV |\Omega\> - \delta^2 |\Omega\> \| \leq \mathcal{O}(\sqrt{g(\ell_{\text{min}})})
\label{uv1}
\end{align}
Likewise,
\begin{align}
\| VU |\Omega\> - \delta^2 |\Omega\> \| \leq \mathcal{O}(\sqrt{g(\ell_{\text{min}})})
\label{uv2}
\end{align}
At the same time, $U$ and $V$ \emph{anti-commute}:
\begin{align}
UV = -VU
\label{uv3}
\end{align}
Indeed, this follows immediately from the fact that $P_3 O_3$ is odd under the symmetry $S = \prod_k S_k$. Combining (\ref{uv1})-(\ref{uv3}), we deduce that
\begin{align*}
\delta^2 \leq \mathcal{O}(\sqrt{g(\ell_{\text{min}})})
\end{align*}
Plugging in the expression for $g$ gives the desired inequality: $\ell_{\text{min}} \leq \tilde{\mathcal{O}}(\log \delta^{-1}/\epsilon)$.
\end{proof}

\begin{reptheorem}{exchconst}
If $|\Omega\>$ has a $(\delta, \ell)$ disorder parameter that is odd under $S$, defined at $i_1, i_2, i_3, i_4$, then
\begin{align*}
\min_{k,l}[\mathrm{dist}(i_k,i_l)] \leq 2\ell + \tilde{\mathcal{O}}\left(\frac{\log \delta^{-1}}{\epsilon} \right)  
\end{align*}
\end{reptheorem}

\begin{proof}
Let $O_1, O_2, O_3, O_4$ be the postulated disorder parameters. As before, we can assume without loss of generality that
\begin{align*} 
\<\Omega| O_{1}^\dagger O_{2} \prod_{k=i_1+1}^{i_2} S_k | \Omega\> = \delta \nonumber \\
\<\Omega| O_{2}^\dagger O_{3} \prod_{k=i_2+1}^{i_3} S_k | \Omega\> = \delta \nonumber \\
\<\Omega| O_{2}^\dagger O_{4} \prod_{k=i_2+1}^{i_4} S_k | \Omega\> = \delta
\end{align*}
Let $\ell_{\text{min}} = \min_{k,l}[\mathrm{dist}(i_k,i_l)] - 2\ell$ and define 
\begin{align*}
U &= P_{2} P_{3} O_{2}^\dagger O_{3} \prod_{k=i_2+1}^{i_3} S_k \nonumber \\
V &= P_{1} P_{2} O_1^\dagger O_2 \prod_{k=i_1+1}^{i_2} S_k \nonumber \\
W &= P_{2} P_{4} O_{2}^\dagger {O}_{4} \prod_{k=i_2+1}^{i_4} S_k
\end{align*}
where $P_i$ are defined in the same way as in (\ref{Pidef}). Then according to Proposition \ref{tildeO},
\begin{align*}
\|U|\Omega\> - \delta|\Omega\>\| &\leq \mathcal{O}(\sqrt{g(\ell_{\text{min}})}) \nonumber \\
\|V|\Omega\> - \delta|\Omega\>\| &\leq \mathcal{O}(\sqrt{g(\ell_{\text{min}})}) \nonumber \\
\|W|\Omega\> - \delta|\Omega\>\| &\leq \mathcal{O}(\sqrt{g(\ell_{\text{min}})})
\end{align*}
It follows that
\begin{align}
\| UVW |\Omega\> - \delta^3 |\Omega\> \| \leq \mathcal{O}(\sqrt{g(\ell_{\text{min}})})
\label{uvw1}
\end{align}
Likewise,
\begin{align}
\| WVU |\Omega\> - \delta^3 |\Omega\> \| \leq \mathcal{O}(\sqrt{g(\ell_{\text{min}})})
\label{uvw2}
\end{align}
At the same time, it is not hard to see that $U,V, W$ obey the ``fermionic'' commutation algebra\cite{Levinwenferm}
\begin{align}
UVW = -WVU 
\label{uvw3}
\end{align}
To see this, note that
\begin{align*}
UVW &= (P_{2} P_{3} O_{2}^\dagger O_{3} \prod_{k=i_2+1}^{i_3} S_k) (P_{1} P_{2} O_1^\dagger O_2 \prod_{k=i_1+1}^{i_2} S_k) (P_{2} P_{4} O_{2}^\dagger {O}_{4} \prod_{k=i_2+1}^{i_4} S_k) \nonumber \\
&= (P_1 O_1^\dagger P_3 O_3 P_4 O_4) (P_{2} O_{2}^\dagger \prod_{k=i_2+1}^{i_3} S_k) (P_{2} O_2 \prod_{k=i_1+1}^{i_2} S_k) (P_{2} O_{2}^\dagger \prod_{k=i_2+1}^{i_4} S_k) \nonumber \\
&= (P_1 O_1^\dagger P_3 O_3 P_4 O_4) (P_{2} O_{2}^\dagger \prod_{k=i_2+1}^{i_4} S_k) (P_{2} O_2 \prod_{k=i_1+1}^{i_2} S_k) (P_{2} O_{2}^\dagger 
\prod_{k=i_2+1}^{i_3} S_k) \nonumber \\
&= - (P_{2} P_{4} O_{2}^\dagger {O}_{4} \prod_{k=i_2+1}^{i_4} S_k) (P_{1} P_{2} O_1^\dagger O_2 \prod_{k=i_1+1}^{i_2} S_k) (P_{2} P_{3} O_{2}^\dagger O_{3} \prod_{k=i_2+1}^{i_3} S_k) \nonumber \\
&= -WVU
\end{align*}
where the minus sign comes from the fact that $P_3 O_3$ is odd under $S$. Combining (\ref{uvw1})-(\ref{uvw3}), we deduce that
\begin{align*}
\delta^3 \leq \mathcal{O}(\sqrt{g(\ell_{\text{min}})})
\end{align*}
Plugging in the expression for $g$ gives the desired inequality: $\ell_{\text{min}} \leq \tilde{\mathcal{O}}(\log(\delta^{-1})/\epsilon)$.
\end{proof}

\section{Discussion}
This work could potentially be extended in several directions. The simplest extension would be to consider spin chains with more general symmetry groups. In particular, we expect that all of our results (Theorems \ref{mainthm1}-\ref{exchconst}) can be easily generalized to arbitrary finite, Abelian symmetry groups. Non-abelian symmetries are more challenging and would be an interesting direction for future work.

Another natural extension would be to consider spin chains whose symmetry transformations are not \emph{onsite}, i.e. not a product of single spin unitaries. (For an example, see the Ising symmetry transformation given in Ref. \cite{ChenLiuWen}). Results about such spin chains would be relevant to the edge physics of two dimensional symmetry protected topological phases.

Finally, one could consider higher dimensional spin systems. We do not know whether Theorem \ref{mainthm1}, which guarantees the existence of order or disorder parameters in gapped, translationally invariant spin chains, has an analog in higher dimensions. However, Theorem \ref{mainthm2}, which guarantees a similar result in the \emph{non-translationally} invariant case, is unlikely to have such a generalization. Indeed, the phenomenon of weak symmetry breaking\cite{WangLevinweak} shows that there exist two dimensional gapped systems with broken symmetries whose order parameters are non-local. These systems do not support either a disorder parameter or a local order parameter anywhere in the weak symmetry breaking region so they provide counterexamples to higher dimensional analogs of Theorem \ref{mainthm2}. 

\section{Acknowledgments}
I would like to thank Maxim Zelenko for pointing out a gap in the proof of Proposition~\ref{tildeO} in a previous draft of this paper and the need for a slightly stronger version of Theorem~\ref{hastthm}. This work was supported in part by the NSF under grant No. DMR-1254741.


\begin{thebibliography}{100}

\bibitem{Fradkin17}
E.~Fradkin, \emph{Disorder operators and their descendants}, J.~Stat.~Phys.~\textbf{167}, p.~427 (2017).

\bibitem{NorbertSPT}
N.~Schuch, D.~Perez-Garcia, and I.~Cirac, \emph{Classifying quantum phases
  using matrix product states and projected entangled pair states}, Phys.~Rev.~B~\textbf{84}, 165139 (2011).

\bibitem{ChenGuWen}
X.~Chen, Z.-C. Gu, and X.-G. Wen, \emph{Complete classification of
  one-dimensional gapped quantum phases in interacting spin systems}, Phys.~Rev.~B~\textbf{84}, 235128 (2011).

\bibitem{HastingsWen}
M.~B. Hastings and X.-G. Wen, \emph{Quasiadiabatic continuation of quantum states: The stability of topological ground-state degeneracy and emergent gauge invariance}, Phys.~Rev.~B~\textbf{72}, 045141 (2005).

\bibitem{Fuchs}
C.~Fuchs and J.~van~de Graaf, \emph{Cryptographic distinguishability measures
  for quantum-mechanical states},  IEEE~Trans.~Inform.~Theory~\textbf{45}, p. 1216 (1999).

\bibitem{Hastings_area_law}
M.~B. Hastings, \emph{An area law for one-dimensional quantum systems}, J.~Stat.~Mech.~Theory~Exp.~\textbf{2007}, p.~8024 (2007).

\bibitem{Hamzaetal}
E.~Hamza, S.~Michalakis, B.~Nachtergaele, and R.~Sims, \emph{Approximating the
  ground state of gapped quantum spin systems}, J.~Math.~Phys.~\textbf{50}, 095213 (2009).

\bibitem{AKLV}
I.~Arad, A.~Kitaev, Z.~Landau, and U.~Vazirani, \emph{An area law and
  sub-exponential algorithm for 1D systems}, arXiv:1301.1162.

\bibitem{thooft78}
G.~'t~Hooft, \emph{On the phase transition towards permanent quark
  confinement}, Nucl.~Phys.~B~\textbf{138}, p.~1 (1978).

\bibitem{levin2013protected}
M.~Levin, \emph{Protected edge modes without symmetry}, Phys.~Rev.~X~\textbf{3}, 021009 (2013).

\bibitem{kitaevtc}
A.~Yu. Kitaev, \emph{Fault-tolerant quantum computation by anyons}, Ann.~Phys.~\textbf{303}, p. 2 (2003).

\bibitem{BravyiKitaev}
S.~B. Bravyi and A.~Y. Kitaev, \emph{Quantum codes on a lattice with boundary},
  quant-ph/9811052 (1998).

\bibitem{Levinprep}
M.~Levin, (in preparation).

\bibitem{Levinwenferm}
M.~Levin and X.-G. Wen, \emph{Fermions, strings, and gauge fields in lattice
  spin models}, Phys.~Rev.~B~\textbf{67}, 245316 (2003).

\bibitem{uhlmann}
A.~Uhlmann, \emph{The transition probability in the state space of a
  *-algebra}, Rep.~Math.~Phys.~\textbf{9}, p.~273 (1976).

\bibitem{Hastings_private}
M.~B. Hastings, (private communication).

\bibitem{Hastings_LSM}
M.~B. Hastings, \emph{Lieb-Schultz-Mattis in higher dimensions}, Phys.~Rev.~B~\textbf{69}, 104431 (2004).

\bibitem{ChenLiuWen}
X.~Chen, Z.-X. Liu, X.-G. Wen, \emph{Two-dimensional symmetry-protected topological orders and their protected gapless edge excitations}, Phys.~Rev.~B~\textbf{84}, 235141 (2011).

\bibitem{WangLevinweak}
C.~Wang and M.~Levin, \emph{Weak symmetry breaking in two-dimensional topological insulators}, Phys.~Rev.~B~\textbf{88}, 245136 (2013).

\end{thebibliography}
\end{document}